\documentclass{theoretics}

\usepackage[T1]{fontenc}

\makeatother

\newcommand{\eps}{\varepsilon}
\newcommand{\ALG}{\mathsf{ALG}}
\newcommand{\RED}{\mathsf{RED}}
\newcommand{\OPT}{\mathsf{OPT}}
\newcommand{\opt}{\mathsf{opt}}
\newcommand{\cost}{\mathsf{cost}}
\newcommand{\credit}{\mathsf{cr}}

\newcommand{\set}[1]{\left\{#1 \right\}}

\newcounter{casei}
\newcounter{caseii}[casei]
\newcounter{caseiii}[caseii]

\newenvironment{caseanalysis}{\setcounter{casei}{0}}{}

\newcommand{\case}[1]{%
\refstepcounter{casei}%
\smallskip\noindent\textbf{\boldmath(\arabic{casei}) #1}\unboldmath%
\phantomsection%
}
\newcommand{\subcase}[1]{%
\refstepcounter{caseii}%
\smallskip\noindent\textbf{\boldmath(\arabic{casei}.\arabic{caseii}) #1}\unboldmath%
\phantomsection%
}
\newcommand{\subsubcase}[1]{%
\refstepcounter{caseiii}%
\smallskip\noindent\textbf{\boldmath(\arabic{casei}.\arabic{caseii}.\arabic{caseiii}) #1}\unboldmath%
\phantomsection%
}

\ThCSauthor[1]{Miguel Bosch-Calvo}{miguel.boschcalvo@idsia.ch}[0000-0002-8647-6928]
\ThCSauthor[1]{Fabrizio Grandoni}{fabrizio.grandoni@idsia.ch}[0000-0002-9676-4931]
\ThCSauthor[2]{Afrouz Jabal Ameli}{a.jabal.ameli@tue.nl}[0000-0001-5620-9039]

\ThCSaffil[1]{IDSIA, USI-SUPSI, Switzerland}
\ThCSaffil[2]{TU Eindhoven, Netherlands}

\ThCSthanks{This article was invited from ICALP 2023 \cite{BGJ23}. The first 2 authors are partially supported by the SNF Grant 200021\_200731 / 1.}
\ThCSshortnames{M. Bosch-Calvo, F. Grandoni, A. Jabal Ameli}
\ThCSshorttitle{A 4/3 Approximation for 2-Vertex-Connectivity}
\ThCSyear{2025}
\ThCSarticlenum{13}
\ThCSreceived{Dec 20, 2023}
\ThCSrevised{Feb 19, 2025}
\ThCSaccepted{Mar 30, 2025}
\ThCSpublished{Jun 16, 2025}
\ThCSdoicreatedtrue
\ThCSkeywords{Combinatorial optimisation, approximation algorithms, graph algorithms, network design}

\title{A \texorpdfstring{$4/3$}{4/3} Approximation for \texorpdfstring{$2$}{2}-Vertex-Connectivity}

\begin{document}
\maketitle

\begin{abstract}
The 2-Vertex-Connected Spanning Subgraph problem (2VCSS) is among the most basic NP-hard (Survivable) Network Design problems: we are given an (unweighted) undirected graph $G$. Our goal is to find a spanning subgraph $S$ of $G$ with the minimum number of edges which is $2$-vertex-connected, namely $S$ remains connected after the deletion of an arbitrary node. 2VCSS is well-studied in terms of approximation algorithms, and the current best (polynomial-time) approximation factor is $10/7$ by Heeger and Vygen [SIDMA'17] (improving on earlier results by Khuller and Vishkin [STOC'92] and Garg, Vempala and Singla [SODA'93]).

Here we present an improved $4/3$ approximation. Our main technical ingredient is an approximation preserving reduction to a conveniently structured subset of instances which are ``almost'' 3-vertex-connected. The latter reduction might be helpful in future work. 
\end{abstract}

\section{Introduction}

Real-world networks are prone to failures. For this reason it is important to design them so that they are still able to support a given amount of traffic despite a few (typically temporary) failures of nodes or edges. The basic goal of survivable network design is to construct cheap networks which are resilient to such failures. 

Most natural survivable network design problems are NP-hard, and a lot of work was dedicated to the design of approximation algorithms for them. One of the most basic survivable network design problems is the 2-Vertex-Connected Spanning Subgraph problem (2VCSS). Recall that an (undirected) graph $G=(V,E)$ is $k$-vertex-connected (kVC) if, after removing any subset $W$ of at most $k-1$ nodes (with all the edges incident to them), the residual graph $G[V\setminus W]$ is connected. In particular, in a 2VC graph $G$ we can remove any single node while maintaining the connectivity of the remaining nodes (intuitively, we can tolerate a single node failure). In 2VCSS we are given a 2VC (unweighted) undirected graph $G=(V,E)$, and our goal is to compute a minimum cardinality subset $S\subseteq E$ of edges such that the (spanning) subgraph $(V,S)$ is 2VC. 

2VCSS is NP-hard: indeed an $n$-node graph $G$ admits a Hamiltonian cycle iff it contains a 2VC spanning subgraph with $n$ edges. Czumaj and Lingas \cite{CL99} proved that the problem is APX-hard, hence it does not have a PTAS unless P=NP. A 2-approximation for 2VCSS can be obtained in different ways. For example one can compute an (open) ear decomposition of the input graph and remove the trivial ears (containing a single edge). The resulting graph is 2VC and contains at most $2(n-1)$ edges (while the optimum solution must contain at least $n$ edges). 
The first non-trivial $5/3$ approximation was obtained by Khuller and Vishkin \cite{KV94}. This was improved to $3/2$ by Garg, Vempala and Singla \cite{GVS93} (see also an alternative $3/2$ approximation by Cheriyan and Thurimella \cite{CT00}). Finally Heeger and Vygen \cite{HV17} presented the current-best $10/7$ approximation\footnote{Before \cite{HV17} a few other papers claimed even better approximation ratios \cite{GR05,JRV03}, however they have been shown to be buggy or incomplete, see the discussion in \cite{HV17}.}. Our main result is as follows (please see Section \ref{sec:overview} for an overview of our approach):   
\begin{theorem}\label{thm:mainTheorem}
There is a polynomial-time $\frac{4}{3}$-approximation algorithm for 2VCSS.    
\end{theorem}

\subsection{Related Work}

An undirected graph $G$ is $k$-edge-connected (kEC) if it remains connected after removing up to $k-1$ edges. The 2-Edge-Connected Spanning Subgraph problem (2ECSS) is the natural edge-connectivity variant of 2VCSS, where the goal is to compute a 2EC spanning subgraph with the minimum number of edges. Like 2VCSS, 2ECSS does not admit a PTAS unless $P=NP$ \cite{CL99}. It is not hard to compute a $2$ approximation for 2ECSS. For example it is sufficient to compute a DFS tree and augment it greedily. Khuller and Vishkin \cite{KV92} found the first non-trivial $3/2$-approximation algorithm. Cheriyan, Seb{\"{o}} and Szigeti \cite{CSS01} improved the approximation factor to $17/12$. This was further improved to $4/3$ in two independent and drastically different works by Hunkenschr{\"o}der, Vempala and Vetta \cite{HVV19} and Seb{\"o} and Vygen \cite{SV14}. The current best and very recent $\frac{118}{89}+\eps<1.326$ approximation is due to Garg, Grandoni and Jabal Ameli \cite{GGJ23}. Our work exploits several ideas from the latter paper. The $k$-Edge Connected Spanning Subgraph problem (kECSS) is the natural generalization of 2ECSS to any connectivity $k\geq 2$ (see, e.g., \cite{CT00,GG12}).

A major open problem in the area is to find a better than 2 approximation for the weighted version of 2ECSS. This is known for the special case with 0-1 edge weights, a.k.a. the Forest Augmentation problem, by the recent work by Grandoni, Jabal-Ameli and Traub \cite{GJT22} (see also \cite{BDS22,CDGKN20,CCDZ20} for the related Matching Augmentation problem).

A problem related to kECSS is the $k$-Connectivity Augmentation problem (kCAP): given a $k$-edge-connected undirected graph $G$ and a collection of extra edges $L$ (\emph{links}), find a minimum cardinality subset of links $L'$ whose addition to $G$ makes it $(k+1)$-edge-connected. It is known~\cite{DKL76} that kCAP can be reduced to the case $k=1$, a.k.a. the Tree Augmentation problem (TAP), for odd $k$ and to the case $k=2$, a.k.a. the Cactus Augmentation problem (CacAP), for even $k$. Several approximation algorithms better than $2$ are known for TAP \cite{A17,CG18,CG18a,EFKN09,FGKS18,GKZ18,KN16,KN16b,N03}, culminating with the current best $1.393$ approximation by Cecchetto, Traub and Zenklusen \cite{CTZ21}. Until recently no better than $2$ approximation was known for CacAP (excluding the special case where the cactus is a single cycle \cite{GGJS19}): the first such algorithm was described by Byrka, Grandoni and Jabal Ameli \cite{BGJ20}, and later improved to $1.393$ in \cite{CTZ21}. In a recent breakthrough by Traub and Zenklusen, a better than $2$ (namely $1.694$) aproximation for the weighted version of TAP was achieved \cite{TZ21} (later improved to $1.5+\eps$ in \cite{TZ23}). Partial results in this direction where achieved earlier in \cite{A17,CN13,FGKS18,GKZ18,N17}. Even more recently, Traub and Zenklusen obtained a $1.5+\eps$ approximation for the weighted version of kCAP \cite{TZ23}.

\subsection{Preliminaries}\label{sec:preliminaries}

We use standard graph notation. For a graph $G=(V,E)$, we  let $V(G)=V$ and $E(G)=E$ denote its nodes and edges, respectively. For $W\subseteq V$ and $F\subseteq E$, we use the shortcuts $G\setminus F:=(V,E\setminus F)$ and $G\setminus W:=G[V\setminus W]$. Throughout this paper we sometimes use interchangeably a subset $F$ of edges and the corresponding subgraph $(W,F)$, $W=\{v\in V: v\in f\in F\}$. The meaning will be clear from the context. For example, we might say that $F\subseteq E$ is 2VC or that $F$ contains a connected component. In particular, we might say that $S\subseteq E$ is a 2VC spanning subgraph. Also, given two subgraphs $G_1$ and $G_2$, by $G'=G_1\cup G_2$ we mean that $G'$ is the subgraph induced by $E(G_1)\cup E(G_2)$. We sometimes represent paths and cycles as sequence of nodes. A $k$-vertex-cut of $G$ is a subset $W$ of $k$ nodes such that $G[V\setminus W]$ has at least 2 connected components. A node defining a $1$-vertex-cut is a cut vertex. By $\OPT(G)\subseteq E(G)$ we denote an optimum solution to a 2VCSS instance $G$, and let $\opt(G):=|\OPT(G)|$ be its size. All the algorithms described in this paper are deterministic.

We will use the notion of block-cutpoint graph. The block-cutpoint graph $G^*$ of a graph $G$ is a bipartite graph that has a node corresponding to every maximal 2VC subgraph of $G$, as well as a node corresponding to every cut vertex of $G$. Furthermore, $G^*$ has an edge $ub$ iff $u$ is a cut vertex and $u\in V(B)$, where $B$ is the maximal 2VC subgraph of $G$ corresponding to $b$ in $G$. It is a well-known fact~\cite[Chapter~3]{west2001introduction} that the block-cutpoint graph of a graph $G$ is a tree if $G$ is connected.

\section{Overview of Our Approach}
\label{sec:overview}

In this section we sketch the proof of our $4/3$-approximation (Theorem \ref{thm:mainTheorem}). The details and proofs which are omitted here will be given in the following technical sections. 

Our result relies on $3$ main ingredients. The first one is an approximation-preserving (up to a small additive term) reduction of 2VCSS to instances of the same problem on properly \emph{structured} graphs, which are ``almost'' 3VC in a sense described later (see Section \ref{sec:overview:structured}).

At this point we compute a minimum-size 2-edge-cover $H$ similarly to prior work: this provides a lower bound on the size of the optimal solution. For technical reasons, we transform $H$ into a \emph{canonical} form,  without increasing its size (see Section \ref{sec:overview:canonical}).

The final step is to convert $H$ into a feasible solution $S$. Starting from $S=H$, this is done by iteratively adding edges to and removing edges from $S$ in a careful manner. In order to keep the size of $S$ under control, we assign $1/3$ credits to each edge of the initial $S$, and use these credits to pay for any increase in the number of edges of $S$ (see Section \ref{sec:overview:credits}). We next describe the above ingredients in more detail.

\subsection{A Reduction to Structured Graphs}
\label{sec:overview:structured}

Our first step is an approximation-preserving (up to a small additive factor) reduction of 2VCSS to instances of the same problem on properly \emph{structured} graphs. This is similar in spirit to an analogous reduction for 2ECSS in \cite{GGJ23}. In particular, we exploit the notion of irrelevant edges and isolating cuts defined in that paper. We believe that our reduction might be helpful also in future work.

In more detail, we can get rid of the following \emph{irrelevant} edges.
\begin{definition}[Irrelevant edge]
Given a graph $G$, we say that an edge $e=uv\in E(G)$ is \emph{irrelevant} if $\{u,v\}$ is a 2-vertex-cut of $G$.
\end{definition}
\begin{restatable}{lemma}{lemirrelevant}\label{lem:irrelevant}
Given a 2VC graph $G$, let $e$ be an irrelevant edge of $G$. Then, every optimal 2VCSS solution for $G$ does not contain $e$.
\end{restatable}
\begin{proof}
Recall that an ear-decomposition of an undirected graph $G$ is a sequence of paths or cycles $P_1,\ldots,P_k$ (ears) spanning $E(G)$ such that $P_1$ is a cycle and $P_i$, $i > 1$, has its internal nodes disjoint from $V_{i-1}:=V(P_1)\cup\ldots\cup V(P_{i-1})$ and its endpoints (or one node if $P_i$ is a cycle) in $V_{i-1}$. We say that an ear-decomposition is open if $P_i$ is a path, for $i > 1$. Every 2VC graph admits an open ear decomposition \cite[Chapter~15]{S03}. We will need the following observation: 
\begin{fact}\label{lem:chord}
Suppose that a minimal solution $S$ to 2VCSS on a graph $G$ contains a cycle $C$. Then $S$ does not contain any chord $f$ of $C$. Indeed, otherwise consider any open ear decomposition of $S$ which uses $C$ as a first ear. Then $f$ would be a trivial ear (consisting of a single edge) of the decomposition, and thus $S\setminus \{f\}$ would also be 2VC, contradicting the minimality of $S$.
\end{fact}

Let $H\subseteq E$ be any optimal (hence minimal) solution to 2VCSS on $G$. Assume by contradiction that $H$ contains an irrelevant edge $e=uv$. Removing $u$ and $v$ splits $H$ into distinct connected components $C_1,\dots, C_k$, with $k\geq 2$. Each one of those components has edges $u_iu$, $v_iv$ in $H$, where $u_i, v_i\in C_i$ for $i\in \set{1,\dots, k}$, otherwise $H$ would contain a cut vertex. Let $P_1$ be a path from $u_1$ to $v_1$ in $C_1$, and $P_2$ be a path from $v_2$ to $u_2$ in $C_2$. Then $e$ is a chord of the cycle $P_1\cup P_2\cup\{uu_1,v_1v,vv_2,u_2u\}$, contradicting the minimality of $H$ by Fact \ref{lem:chord}.
\end{proof}
We can enforce (see later) that our graph $G$ is ``almost'' 3VC, in the sense that the only 2-vertex-cuts of $G$ are a very specific type of \emph{isolating} cuts defined as follows.
\begin{restatable}[Isolating cut]{definition}{defisolating}
Given a 2-vertex-cut $\{u,v\}$ of a graph $G$, we say that this cut is \emph{isolating} if $G\setminus\set{u,v}$ has exactly two connected components, one of which consists of $1$ node. Otherwise the cut is \emph{non-isolating}.
\end{restatable}
Assuming that there are no non-isolating cuts, we can avoid the following local configuration: this will be helpful in the rest of our analysis.
\begin{restatable}[Removable 5-cycle]{definition}{def:removable}
    We say that a 5-cycle $C$ of a 2VC graph $G$ is \emph{removable} if it has at least two vertices of degree 2 in $G$.
\end{restatable}
\begin{restatable}{lemma}{lemremovable}\label{lem:removable}
    Given a 2VC graph $G$ without non-isolating cuts and with at least $6$ nodes. Let $C$ be a removable 5-cycle of $G$. Then in polynomial time one can find an edge $e$ of $C$ such that there exists an optimum solution to 2VCSS on $G$ not containing $e$ (we say that $e$ is a \emph{removable} edge).
\end{restatable}
\begin{proof}
    Assume $C=v_1v_2v_3v_4v_5$.
    If $C$ has two vertices of degree $2$ that are adjacent in $C$, namely $v_1$ and $v_2$, then $\{v_3,v_5\}$ is a non-isolating cut of $G$, a contradiction. Thus, we can assume that $C$ has exactly two non-adjacent vertices of degree $2$, say $v_1$ and $v_3$ w.l.o.g. 
    
    We will show that the edge $e=v_4v_5$ is the desired removable edge. Let $H$ be an optimal 2VCSS solution for $G$ that uses the edge $v_4v_5$. Observe that in this case since $v_1$ and $v_3$ have degree $2$, then $H$ must contain all the edges of $C$.
    
    To complete the argument we show that there exists an edge $f\in E(G)\setminus E(H)$ such that $v_4v_5$ is a chord of a cycle in $H':=H\cup\{f\}$: hence we can remove $v_4v_5$ from $H'$     
    using Fact~\ref{lem:chord} 
    to obtain an alternative optimum solution not containing $v_4v_5$. 
    
    Let $H''=H\setminus\{v_4v_5\}$. There is no cycle $C'$ in $H''$ that contains both $v_4$ and $v_5$, otherwise $v_4v_5$ is a chord of $C'$ in $H$, contradicting the minimality of $H$ by Fact \ref{lem:chord}. Therefore if we remove $v_2$ from $H''$, there must be no paths from $v_4$ to $v_5$. This means that there is a partition of $V(G)\setminus\{v_2\}$ into non-empty sets $V_1$ and $V_2$ such that, $\{v_3,v_4\}\in V_1$, $\{v_1,v_5\}\in V_2$ and there is no edge in $H''$ between $V_1$ and $V_2$.  Since $|V(G)|\ge 6$, then we can assume w.l.o.g that $|V_1|\ge3$.
    
    There must be an edge $f=u_1u_2\in E(G)$ such that $u_1\in V_1\setminus\{v_3, v_4\}$ and $u_2 \in V_2$, otherwise  $\{v_2,v_4\}$ is a non-isolating cut in $G$, a contradiction. 
    Now we show that $f$ is the desired edge. We claim that there exists a path $P_1$ in $H[V_1\setminus \{v_3\}]$ between $u_1$ and $v_4$. Since $H$ is 2VC, there exists a path $P_1$ between $u_1$ and $v_4$ not using $v_2$. Such path does not use $v_3$ either since this node is adjacent only to $v_2$ and $v_4$, and $u_1\notin \{v_3,v_4\}$. If $P_1$ is not contained in $H[V_1]$, it would need to use at least two edges between $V_1$ and $V_2$ in $H$. However, since there is no edge in $H''$ between $V_1$ and $V_2$, the only edge in $H$ between $V_1$ and $V_2$ is $v_4v_5$, so $P_1$ cannot use two edges between $V_1$ and $V_2$, and thus it must be contained in $H[V_1]$.
    
    Symmetrically, we claim that there exists a path $P_2$ in $H[V_2\setminus \{v_1\}]$ between $u_2$ and $v_5$. Notice that $u_2=v_5$ is possible, in which case the claim trivially holds. Hence, next assume $u_2\neq v_5$. Observe that $u_2\neq v_1$ since $u_2$ is adjacent to $u_1\notin \{v_2,v_5\}$. Thus, the claim about $P_2$ follows symmetrically to the case of $P_1$. Altogether, $v_4v_5$ is a chord of the cycle $P_1\cup P_2 \cup \{f\} \cup C \setminus\{v_4v_5\}$ in $H'=H\cup\{f\}$, which implies the lemma.
\end{proof}

\begin{figure}
    \centering
    \includegraphics{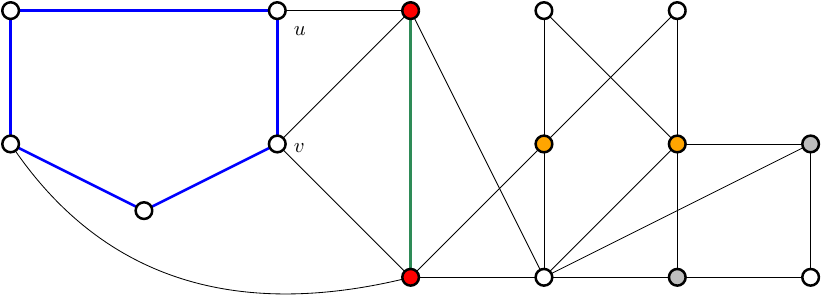}
    \caption{The cycle induced by the blue edges is a removable cycle, since it has two vertices of degree $2$ in $G$. The edge $uv$ is removable. The red and orange (resp. gray) pairs of vertices form a non-isolating (resp. isolating) cut. The green edge is irrelevant.}
    \label{fig:structured}
\end{figure}

We are now ready to define a structured graph and to state our reduction to such graphs (see also Figure~\ref{fig:structured}).
\begin{restatable}[Structured graph]{definition}{defstructured}
    A 2VC graph $G$ is structured if it does not contain: (1) Irrelevant edges; (2) Non-isolating cuts, and (3) Removable $5$-cycles.
\end{restatable}
\begin{restatable}{lemma}{lemreduction}\label{lem:reduction}
    Given a constant $1<\alpha\leq \frac{3}{2}$, if there exists a polynomial-time algorithm for 2VCSS on a structured graph $G$ that returns a solution of cost at most $\max\{\opt(G),\alpha\cdot\opt(G) - 2\}$, then there exists a polynomial-time $\alpha$-approximation algorithm for 2VCSS.
\end{restatable}
We remark that any $\alpha-\eps$ approximation of 2VCSS on structured graphs, for an arbitrarily small constant $\eps>0$, immediately implies an algorithm of the type needed in the claim of Lemma~\ref{lem:reduction}: indeed, instances with $\opt(G)\leq \max\{\frac{2}{\eps},\frac{2}{\alpha-1}\}$ can be solved exactly in constant time by brute force.

\begin{algorithm}
\caption{Algorithm $\RED{}$ that reduces from arbitrary to structured instances of 2VCSS. Here $G$ is 2VC and $\ALG$ is an algorithm for structured instances that returns a solution of cost at most $\max\{\opt(G),\alpha\cdot\opt(G) - 2\}$ for some $1<\alpha\leq \frac{3}{2}$.}
\label{fig:algo}
\KwIn{A 2-vertex-connected graph $G$}
\KwOut{A subgraph $H$ of $G$}
\If{$|V(G)|< \max\{6,\frac{2}{\alpha-1}\}$}{\label{line:bruteForceCondition}
    Compute $\OPT(G)$ by brute force (in constant time) and \textbf{return} $\OPT(G)$\;\label{line:bruteForce}
}
\If{$G$ contains an irrelevant edge}{\label{line:irrelevantCondition}
    \textbf{return} $\RED(G \setminus \{e\})$\;\label{line:irrelevant}
}
\If{$G$ contains a non-isolating vertex cut $\set{u, v}$}{\label{line:nonIsolatingCondition}
    Let $(V_1,V_2), 2\leq|V_1|\leq|V_2|$, be a partition of $V(G)\setminus\set{u,v}$ such that there are no edges between $V_1$ and $V_2$ in $G\setminus\set{u,v}$\;
    Let $G_1$ be the graph resulting from $G$ by contracting $V_2$ into one node $v_2$ and $G_2$ by contracting $V_1$ into one node $v_1$ (keeping one copy of parallel edges in both cases)\;
    Let $H_1=\RED(G_1)$ and $H_2=\RED(G_2)$\;
    Let $E_1$ (resp. $E_2$) be the two edges of $H_1$ (resp., $H_2$) with endpoints in $v_2$ (resp., $v_1$)\;
    \textbf{return} $H := (H_1 \setminus E_1) \cup (H_2 \setminus E_2)$\;\label{line:nonIsolating}
}
\If{$G$ contains a removable $5$-cycle}{\label{line:removable5CycleCondition}
    Let $e$ be the removable edge (found via Lemma~\ref{lem:removable}) in that cycle and \textbf{return} $\RED(G\setminus \{e\})$\;\label{line:removable5Cycle}
}
\textbf{return} $\ALG(G)$\;\label{line:baseCase}
\end{algorithm}

The algorithm at the heart of Lemma~\ref{lem:reduction} is algorithm $\RED$ given in Figure \ref{fig:algo}. Lines~\ref{line:bruteForceCondition}-\ref{line:bruteForce} solve instances with few new nodes by brute force. Lines~\ref{line:irrelevantCondition}-\ref{line:irrelevant}, \ref{line:nonIsolatingCondition}-\ref{line:nonIsolating}, and \ref{line:removable5CycleCondition}-\ref{line:removable5Cycle} get rid recursively of irrelevant edges, non-isolating vertex cuts and removable 5-cycles, resp. When Line~\ref{line:baseCase} is reached, the graph is structured and therefore we can apply a black-box algorithm $\ALG$ for structured instances of 2VCSS. 

It is easy to see that the algorithm runs in polynomial time. 
\begin{lemma}\label{lem:alg:time}
    $\RED(G)$ runs in polynomial time in $|V(G)|$ if $\ALG$ does so.
\end{lemma}
\begin{proof}
    Let $n=|V(G)|$. First observe that each recursive call, excluding the corresponding subcalls, can be executed in polynomial time. In particular,     
    we can find one irrelevant edge, if any, in polynomial time by enumerating all the possible 2-vertex-cuts. Furthermore, we can find some removable $5$-cycle, if any, in polynomial time by enumerating all 5-cycles. Then, by Lemma~\ref{lem:removable}, we can indentify a removable edge in such cycle. We also remark that in Lines~\ref{line:irrelevant} and~\ref{line:removable5Cycle} we remove one edge, and we never increase the number of edges. Hence the corresponding recursive calls increase the overall running time by a polynomial factor altogether. 

    It is then sufficient to bound the number $f(n)$ of recursive calls where we execute Lines~\ref{line:nonIsolatingCondition}-\ref{line:nonIsolating} starting from a graph with $n$ nodes. Consider one recursive call on a graph $G$ with $n$ nodes, where the corresponding graph $G_1$ has $5\leq k\leq n/2+2$ nodes. Notice that $G_2$ has $n-k+4$ nodes. Thus, one has $f(n)\leq \max_{5\leq k\leq n/2+2}\{f(k)+f(n-k+4)\}$, which implies that $f(n)$ is polynomially bounded.
\end{proof}
Let us next show that $\RED$ produces a feasible solution. 
\begin{lemma}\label{lem:redFeasible}
    Given a 2VC graph $G$, $\RED(G)$ returns a feasible 2VCSS solution for $G$.
\end{lemma}
\begin{proof}
 Let us prove the claim by induction on $|E(G)|$. The base cases are given when $\RED(G)$ executes Lines~\ref{line:bruteForce} or~\ref{line:baseCase}: in these cases $\RED$ clearly returns a feasible solution. Consider an instance $G$ where $\RED(G)$ does not execute those lines (in the root call), and assume the claim holds for any instance $G'$ where $|E(G')| < |E(G)|$.
 By Lemma \ref{lem:irrelevant}, when $\RED$ recurses at  Line~\ref{line:irrelevant}, the graph $G\setminus \{e\}$ is 2VC, hence the recursive call returns a 2VC spanning subgraph by inductive hypothesis. A similar argument holds when Line~\ref{line:removable5Cycle} is executed, this time exploiting Lemma \ref{lem:removable}.

It remains to consider the case when Lines~\ref{line:nonIsolatingCondition}-\ref{line:nonIsolating} are executed. Let us first prove that $G_1$ and $G_2$ are 2VC. We prove it for $G_1$, the proof for $G_2$ being symmetric. Assume to get a contradiction that $G_1$ has a cut vertex $w$. There must be a path in $G[V_1\cup\{u, v\}]$ between $u$ and $v$, otherwise $G$ is not 2VC, so $w\neq v_2$. If $w\in \{u, v\}$, then $w$ is also a cut vertex in $G$, a contradiction. Therefore $w\in V_1$. If one of the components resulting from removing $w$ from $G_1$ contains neither $u$ nor $v$ then $w$ is also a cut vertex in $G$, a contradiction. Thus, removing $w$ from $G_1$ yields two connected components $C_u, C_v$, with $u\in C_u, v\in C_v$. But the path $uv_2v$ is still present in $G_1\setminus\{w\}$, contradicting the fact that $w$ is a cut vertex in $G_1$.

Notice that $|E(G_1)|, |E(G_2)| < |E(G)|$. Since $G_1$ and $G_2$ are 2VC, we can assume by inductive hypothesis that both $H_1$ and $H_2$ are 2VC. It is left to show that $H$ is 2VC.

Assume to get a contradiction that $H$ has a cut vertex $w$. If $w\in\{u, v\}$, then $w$ is also a cut vertex in either $H_1$ or $H_2$. Thus we can assume w.l.o.g. $w\in V_1$. Consider the components resulting from removing the vertex $w$ from $H$. If one of these components contains neither $u$ nor $v$ then $w$ is also a cut vertex in $H_1$. Thus removing $w$ from $H$ yields two connected components $C_u, C_v$, with $u\in C_u, v\in C_v$. But since $w\in V_1$, no edge from $H_2$ present in $H$ is removed by deleting $w$. In particular, there is a path from $u$ to $v$ in $H$, contradicting the fact that $w$ is a cut vertex.
\end{proof}
It remains to analyze the approximation factor of $\RED$. 
\begin{lemma}
$|\RED(G)|\leq\begin{cases}
    \opt(G), \quad & \text{if } |V(G)|< \max\{6,\frac{2}{\alpha-1}\};\\
    \alpha\cdot \opt(G) - 2, \quad & \text{if } |V(G)|\geq \max\{6,\frac{2}{\alpha-1}\}.
\end{cases}$
\end{lemma}
\begin{proof}
We prove the claim by induction on $|E(G)|$. The base cases correspond to the execution of Lines~\ref{line:bruteForce} and~\ref{line:baseCase}. Here the claim trivially holds. The claim holds by inductive hypothesis and by Lemmas \ref{lem:irrelevant} and \ref{lem:removable} when Lines~\ref{line:irrelevant} and~\ref{line:removable5Cycle}, resp., are executed. Notice that the $6$ that appears in the $\max$ in the claim of the lemma is meant to guarantee that the conditions of Lemma \ref{lem:removable} are satisfied.

It remains to consider the case when Lines~\ref{line:nonIsolatingCondition}-\ref{line:nonIsolating} are executed. Let $\OPT$ be a minimum 2VC spanning subgraph of $G$, and $\OPT_i$ be an optimal 2VCSS solution for $G_i$, $i\in\set{1, 2}$. We will later show:\footnote{In fact, it is not difficult to show that \eqref{eqn:red_opt} holds with equality, but it is not needed for our purposes and therefore we omit it.}
\begin{equation}\label{eqn:red_opt}
|\OPT| \geq |\OPT_1| + |\OPT_2| - 4
\end{equation}

Notice that, since $G$ contains no irrelevant edges, $H_1\setminus E_1$ and $H_2\setminus E_2$ are edge-disjoint. Moreover,  $|H_i\cap E_i| = 2$ for $i\in\set{1, 2}$, so we have $|H|= |H_1|+ |H_2| - 4$. Also, for $|V_i|\geq \frac{2}{\alpha-1}$, one has $|\OPT_i|\leq \alpha|\OPT_i|-2$, by the induction hypothesis. We now distinguish a few cases.

If $|V_2|< \max\{6,\frac{2}{\alpha-1}\}$, then $|H| = |H_1|+ |H_2| - 4 = |\OPT_1| + |\OPT_2| - 4 \leq |\OPT|$. 

If $|V_1| \geq \max\{6,\frac{2}{\alpha-1}\}$, then $|H| = |H_1|+ |H_2| - 4\leq \alpha|\OPT_1| - 2 + \alpha|\OPT_2| - 2  - 4 = \alpha(|\OPT_1| + |\OPT_2|) - 8 \leq \alpha |\OPT| + 4\alpha - 8 \leq \alpha |\OPT| - 2$. The last inequality uses the assumption $\alpha\leq 3/2$.

Finally, if $|V_1|< \max\{6,\frac{2}{\alpha-1}\}$ and $|V_2| \geq \max\{6,\frac{2}{\alpha-1}\}$, we have $|H| = |H_1|+ |H_2| - 4\leq |\OPT_1| + \alpha|\OPT_2| - 2 - 4 = (1-\alpha)|\OPT_1| + \alpha(|\OPT_1| + |\OPT_2|) - 6\leq (1-\alpha)|\OPT_1| + 4\alpha + \alpha |\OPT| - 6\leq \alpha |\OPT| - 2$. The last inequality holds since $|\OPT_1|\geq|V(G_1)|\geq 5$ and $\alpha > 1$.

It remains to prove \eqref{eqn:red_opt}. Assume by contradiction that $|\OPT|< |\OPT_1| + |\OPT_2| - 4$. Notice that, since $G$ contains no irrelevant edges, $E(G)=(E(G_1)\setminus E_1)\dot\cup (E(G_2)\setminus E_2)$ and thus $\OPT=((E(G_1)\setminus E_1)\cap \OPT)\dot\cup((E(G_2)\setminus E_2)\cap\OPT)$. Thus we have that either $|(E(G_1)\setminus E_1)\cap\OPT|<|\OPT_1|-2$ or $|(E(G_2)\setminus E_2)\cap\OPT|<|\OPT_2|-2$. Assume w.l.o.g. that $|(E(G_1)\setminus E_1)\cap\OPT|<|\OPT_1|-2$. Then $((E(G_1)\setminus E_1)\cap\OPT)\cup\set{uv_2,vv_2}$ is a 2VC spanning subgraph of $G_1$ of cardinality less than $\OPT_1$, a contradiction. \eqref{eqn:red_opt} follows.
\end{proof}

\subsection{A Canonical 2-Edge-Cover}
\label{sec:overview:canonical}

It remains to give a good enough approximation algorithm for structured graphs. The first step in our algorithm (similarly to prior work on related problems \cite{CDGKN20,GGJ23,HVV19}) is to compute (in polynomial time \cite[Chapter~30]{S03}) a minimum-cardinality 2-edge-cover\footnote{A 2-edge-cover $H$ of a graph $G$ is a subset of edges such that each node $v$ of $G$ has at least $2$ edges of $H$ incident to it.} $H$ of $G$. It is worth to remark that $|H|\leq \opt(G)$: indeed the degree of each node in any 2VC spanning subgraph of $G$ must be at least $2$. 

For technical reasons, we transform $H$, without increasing its size, into another 2-edge-cover which is \emph{canonical} in the following sense. We need some notation first. If a connected component of $H$ has at least $6$ edges we call it a \emph{large component}, and otherwise a \emph{small component}. Let $C$ be a large component of $H$. We call every maximal 2VC subgraph of $C$ with at least $3$ nodes a \emph{block}, and every edge of $C$ such that its removal splits that component into two connected components a \emph{bridge}. Notice that every edge of $C$ is either a bridge or belongs to some block in that component. Also, every edge of $C$ belongs to at most one block, thus there is a unique partition of the edges of $C$ into blocks and bridges (but a node of $C$ might belong to multiple blocks and to multiple bridges). Observe that $C$ is 2VC iff it has exactly one block. If $C$ is large but not 2VC we call it a \emph{complex component}. If a block $B$ of a complex component $C$ contains only one cut vertex of $C$, we say that $B$ is a \emph{leaf-block} of $C$. An example showcasing the notation we use can be found in Figure~\ref{fig:terminology}.

\begin{figure}
    \centering
    \includegraphics{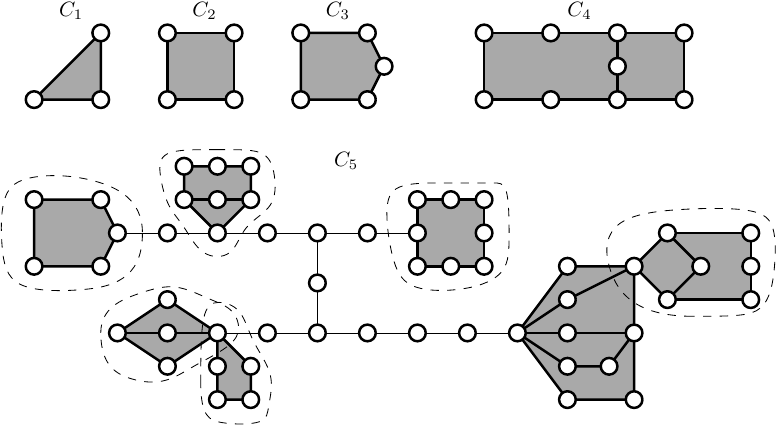}
    \caption{Gray regions represent blocks of $H$. $C_1, C_2$ and $C_3$ are small components of $H$, and $C_4$ and $C_5$ are large components. Also, $C_5$ is a complex component. The non-bold edges are bridges of $C_5$. The dashed lines encompass the leaf-blocks of $C_5$. Notice that $C_4$ contains a single block, while $C_5$ contains many of them.}
    \label{fig:terminology}
\end{figure}

Since $H$ is a $2$-edge-cover, every complex component $C$ must have at least $2$ leaf-blocks. Indeed, take the block-cutpoint graph $C^*$ of $C$. Since $C$ is connected, $C^*$ is a tree, and since $C$ is not 2VC, $C^*$ has at least two leafs. Those leafs cannot correspond to a bridge on $C$ because $H$ is a $2$-edge-cover, so they must correspond to leaf-blocks.

\begin{definition}[Canonical 2-Edge-Cover]
A 2-edge-cover $S$ of a graph $G$ is \emph{canonical} if: (1) Every small component of $S$ is a cycle; (2) For any complex component $C$ of $S$, each leaf-block $B$ of $C$ has at least $5$ nodes.
\end{definition}
\begin{restatable}{lemma}{lemcanonical}\label{lem:canonical}
Given a minimum 2-edge-cover $H$ of a structured graph $G$, in  polynomial time one can compute a \emph{canonical} 2-edge-cover $S$ of $G$ with $|S|=|H|$. 
\end{restatable}
\begin{proof}
We start with $S:=H$. At each step if there are edges $e\in E(G)\setminus E(S)$ and  $e'\in E(S)$, such that $S':=S\cup \{e\}\setminus\{e'\}$ is a $2$-edge-cover that has fewer connected components than $S$, or it has the same number of connected components as $S$ but has fewer bridges and blocks in total than $S$, then we replace $S$ by $S'$. This process clearly terminates within a polynomial number of steps, returning a 2-edge-cover $S$ of the same size as the initial $H$ (hence in particular $S$ must be minimal).

Let us show that the final $S$ satisfies the remaining properties. Assume by contradiction that $S$ has a connected component $C$ with at most $5$ edges that is not a cycle. By a simple case analysis $C$ must be a $4$-cycle plus one chord $f$. However, this contradicts the minimality of $S$.

Finally assume by contradiction that $S$ has a complex component $C$, with a leaf-block $B$ such that $B$ has at most $4$ nodes. By the minimality of $S$, $B$ must be a $3$-cycle or a $4$-cycle. Let $B=v_1\ldots v_k$, $k\in \{3,4\}$, and assume w.l.o.g. that $v_1$ is the only cut-vertex of $C$ that belongs to $B$. In this case we show that there must exist an edge $e=uz\in E(G)$ such that $u\in \{v_2,v_k\}$ and $z\notin B$. If this is not true then for $k=3$, $v_1$ is a cut-vertex in $G$, and for $k=4$, $\{v_1,v_3\}$ form a non-isolating cut, leading to a contradiction in both cases.  Consider $S':=S\cup \{e\}\setminus\{uv_1\}$. Note that $S'$ is a $2$-edge-cover of the same size as $S$. Since $uv_1$ belongs to a cycle of $S$, then the number of connected components in $S'$ is not more than in $S$. If $z\notin C$ the number of connected components of $S'$ is less than in $S$, which is a contradiction. Otherwise, the number of connected components of $S$ and $S'$ is the same. Now in $S'$ all the bridges and the blocks of $S$ that shared an edge with any path from $u$ to $z$ in $S\setminus\{uv_1\}$ become part of the same block and all the other bridges and blocks remain the same. This is a contradiction as the total number of bridges and blocks of $S'$ is less than in $S$. 
\end{proof}

\subsection{A Credit-Based Argument}
\label{sec:overview:credits}

Next assume that we are given a minimum-cardinality canonical 2-edge-cover $H$ of a structured graph $G$. Observe that, for $|H|\leq 5$, $H$ is necessarily a cycle of length $|H|$ by the definition of canonical 2-edge-cover. In particular $H$ is already a feasible (and optimal) solution. Therefore, we next assume $|H|\geq 6$. Starting from $S=H$, we will gradually add edges to (and sometimes remove edges from) $S$, until $S$ becomes 2VC (see Section \ref{sec:2VCSS}). In order to keep the size of $S$ under control, we use a credit-based argument similarly to prior work \cite{CDGKN20,GGJ23,GJT22}. At high level, the idea is to assign a certain number of credits $\credit(S)$ to $S$. Let us define the cost of $S$ as $\cost(S)=|S|+\credit(S)$. We guarantee that for the initial value of $S$, namely $S=H$, $\cost(S)\leq \frac{4}{3}|H|$. Furthermore, during the process $\cost(S)$ does not increase. 

During the process we maintain the invariant that $S$ is canonical. Hence, the following credit assignment scheme is valid for any intermediate $S$:
\begin{enumerate}\itemsep0pt
    \item To every small component $C$ of $S$ we assign $\credit(C)=|E(C)|/3$ credits.
    \item Each large component $C$ receives $\credit(C)=1$ credits.
    \item Each block $B$ receives $\credit(B)=1$ credits.
    \item Each bridge $b$ receives $\credit(b)=1/4$ credits.
\end{enumerate}
An example of such assignment can be seen in Figure~\ref{fig:creditAssignment}. We remark that each large connected component $C$ of $S$ which is 2VC receives one credit in the role of a component, and one additional credit in the role of a block of that component. Let $\credit(S)\geq 0$ the total number of credits assigned to the subgraphs of $S$. It is not hard to show that the initial cost of $S$ is small enough.
\begin{figure}
    \centering
    \includegraphics{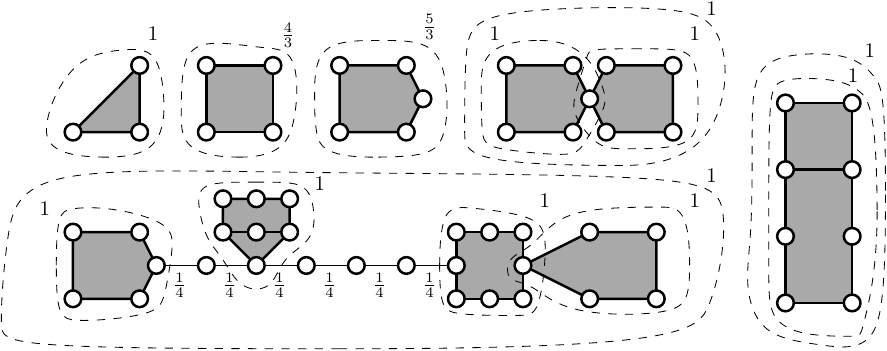}
    \caption{An example of a valid credit assignment. Gray regions represent blocks of $S$. The dashed lines encompass the objects that receive credit (other than bridges). We remark that a large 2VC component receives $2$ credits, one credit in the role of a component, and one additional credit in the role of a block of that component.}
    \label{fig:creditAssignment}
\end{figure}
\begin{lemma}\label{lem:initialCost}
$\cost(H)\leq \frac{4}{3}|H|$.
\end{lemma}
\begin{proof}
Let us initially assign $\frac{1}{4}$ credits to the bridges of $H$ and $\frac{1}{3}$ credits to the remaining edges. Hence, we assign at most $\frac{|H|}{3}$ credits in total. We next redistribute these credits to satisfy the credit assignment scheme.

Each small component $C$ retains the credits of its edges. If $C$ is large and 2VC then it has exactly one block $B$. Since $|E(C)|\geq 6$, its edges have at least $2$ credits, so we can assign $1$ credit to $C$ and $1$ to $B$.

Now consider a complex component $C$ of $H$. The bridges keep their own credits. Since $H$ is a $2$-edge-cover and $C$ is complex, then $C$ has at least $2$ leaf-blocks $B_1$ and $B_2$. By the definition of canonical, $B_1$ and $B_2$ have at least $5$ nodes (hence edges) each. Therefore, together they have at least $\frac{10}{3}>3$ credits, which is sufficient to assign one credit to $C$, $B_1$ and $B_2$.
Any other block $B$ of $C$ (which has at least $3$ edges) keeps the credits of its edges, hence at least $1$ credit.  Observe that $\cost(H)=|H|+\credit(H)\le \frac{4}{3}|H|$ as desired.    
\end{proof}
As mentioned before, starting from $S=H$, we will transform $S$ without increasing its cost $\cost(S)$ until it becomes a single large component $C$ that is 2VC (and thus it has exactly one block $B$) and therefore a 2VC spanning subgraph of $G$. Notice that at the end of the process $\credit(S)=\credit(C)+\credit(B)=2$, hence $|S|= \cost(S)-2\leq \frac{4}{3}|H|-2$. Combining this with the trivial case for $|H|\leq 5$, we obtain the following lemma. 
\begin{restatable}{lemma}{lemmainLemma}\label{lem:mainLemma}
Given a canonical minimum 2-edge-cover $H$ of a structured graph $G$, one can compute in polynomial time a 2VCSS solution $S$ for $G$ with $|S|\leq \max\{|H|,\frac{4}{3}|H| - 2\}$.
\end{restatable}

Given the above results, it is easy to prove Theorem \ref{thm:mainTheorem}.
\begin{proof}[Proof of Theorem \ref{thm:mainTheorem}]
By Lemma \ref{lem:reduction} it is sufficient to compute a solution of cost at most $\max\{\opt(G),\frac{4}{3}\cdot\opt(G) - 2\}$ on a structured graph $G$. We initially compute a canonical minimum $2$-edge-cover $H$ of $G$ via Lemma \ref{lem:canonical}.
Then we apply Lemma \ref{lem:mainLemma} to obtain a 2VCSS solution $S$ with $|S|\leq \max\{|H|,\frac{4}{3}|H|-2\}\leq \max\{\opt(G),\frac{4}{3}\opt(G)-2\}$. Clearly all steps can be performed in polynomial time. 
\end{proof}

It remains to discuss the proof of Lemma \ref{lem:mainLemma} (assuming $|H|\geq 6$), which is the most technical part of our paper. The construction at the heart of the proof consists of a few stages. Recall that we start with a 2-edge-cover $S=H$, and then gradually transform $S$ without increasing $\cost(S)$. We now define a specific type of $4$-cycles that need special care.

\begin{definition}[pendant $4$-cycle]
    Let $S$ be a 2-edge-cover of a graph $G$ and $C'$ be a large component of $S$. We say that a connected component $C$ of $S$ is a \emph{pendant $4$-cycle} (of $C'$) if $C$ is a $4$-cycle and all the edges of $G$ with exactly one endpoint in $V(C)$ have the other endpoint in $V(C')$.
\end{definition}

In the initial phase of our construction (see Section \ref{sec:removeSmall}), we eliminate all small components from $S$, except for pendant $4$-cycles. Pendant $4$-cycles are harder to remove, and thus they require separate arguments. They will be dealt with in later stages of the algorithm.

\begin{restatable}{lemma}{lemsmallComponents}\label{lem:smallComponents}
    Let $G$ be a structured graph and $H$ be a canonical minimum $2$-edge cover of $G$, with $|H|\geq 6$. In polynomial time one can compute a canonical $2$-edge-cover $S$ of $G$ such that the only small components of $S$ are pendant $4$-cycles and $\cost(S)\leq \cost(H)$.
\end{restatable}
In the second stage of our construction (see Section \ref{sec:removeComplex}) we reduce to the case where $S$ consists of large 2VC components only.

\begin{restatable}{lemma}{lemcomplexComponents}\label{lem:complexComponents}
    Let $G$ be a structured graph and $S$ be a canonical $2$-edge-cover of $G$ such that the only small components of $S$ are pendant $4$-cycles. In polynomial time one can compute a canonical 2-edge-cover $S'$ of $G$ such that all the connected components of $S'$ are 2VC and large, and $\cost(S')\leq \cost(S)$.
\end{restatable}
At this point we can exploit the following definition and lemma from \cite{GGJ23} (see Figure~\ref{fig:niceCycle}) to construct the desired 2VC spanning subgraph.
\begin{definition}[Nice Cycle]\label{def:nicePathCycle}
Let $\Pi=(V_1,\ldots,V_k)$, $k\geq 2$, be a partition of the node-set of a graph $G$. A nice cycle $N$ of $G$ w.r.t. $\Pi$ is a subset of edges with endpoints in distinct subsets of $\Pi$ such that: (1) $N$ induces one cycle of length at least $2$ in the graph obtained from $G$ by collapsing each $V_i$ into a single node; (2) given the two edges of $N$ incident to some $V_i$, these edges are incident to distinct nodes of $V_i$ unless $|V_i|=1$.  
\end{definition} 

\begin{figure}
    \centering
    \includegraphics{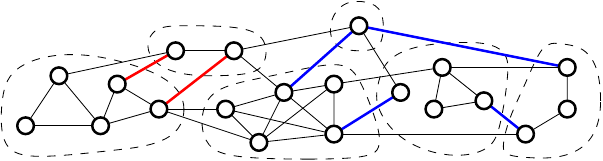}
    \caption{An example of nice cycles w.r.t. the partition $\Pi$ (induced by the dashed lines). In the figure two distinct nice cycles are shown, given by the red and blue edges, respectively.}
    \label{fig:niceCycle}
\end{figure}

\begin{lemma}\label{lem:NiceCycle}\cite{GGJ23}
Let $\Pi=(V_1,\ldots,V_k)$, $k\geq 2$, be a partition of the node-set of a 2VC graph $G$. In polynomial time one can compute a nice cycle $N$ of $G$ w.r.t. $\Pi$. 
\end{lemma}

\begin{restatable}{lemma}{lemlargeComponents}\label{lem:largeComponents}
    Let $G$ be a structured graph and $S$ be a $2$-edge-cover of $G$ such that all the connected components of $S$ are 2VC and large. In polynomial time one can compute a 2VCSS solution $S'$ for $G$ with $\cost(S')\leq \cost(S)$.
\end{restatable}
\begin{proof}
    Initially set $S'=S$. Consider the partition $\Pi=(V_1,\ldots,V_k)$ of $V(G)$ where $V_i$ is the set of vertices of the 2VC component $C_i$ of $S'$. If $k=1$, $S'$ already satisfies the claim. Otherwise, using Lemma~\ref{lem:NiceCycle} we can compute a nice cycle $N$ of $G$ w.r.t. $\Pi$. Let us replace $S'$ with $S'':=S'\cup N$. W.l.o.g  assume $N$ is incident to $V_1,...,V_r$ for some $2\le r\le k$. Then in $S''$ the nodes   $V_1\cup \ldots \cup V_r$ belong to a unique (large) 2VC connected component $C'$. Furthermore $\cost(S')-\cost(S'')=\sum_{i=1}^r(\credit(C_i)+\credit(B_i))-\credit(C')-\credit(B')-r=2r-2-r\ge 0$, where $B_i$ is the only block of the component $C_i$ and $B'$ the only block of $C'$. By iterating the process for a polynomial number of times one obtains a single 2VC component, hence the claim.
\end{proof}

The proof of Lemma \ref{lem:mainLemma} follows by chaining Lemmas \ref{lem:smallComponents}, \ref{lem:complexComponents}, and \ref{lem:largeComponents}, and by the previous simple observations.

\section{From a Canonical 2-Edge-Cover to a 2VC Spanning Subgraph}
\label{sec:2VCSS}

In the following two subsections we provide the proof of Lemmas \ref{lem:smallComponents} and \ref{lem:complexComponents}, hence completing the proof of Lemma \ref{lem:mainLemma}. Throughout this section, we say that two connected components $C_1$ and $C_2$ of $S$ are \emph{adjacent} if there are nodes $u_1\in C_1, u_2\in C_2$ such that $u_1u_2\in E(G)$.

We will use following lemma from \cite{GGJ23} multiple times, whose proof is duplicated here for the sake of completeness.

\begin{restatable}[3-Matching Lemma]{lemma}{3MatchingLemma}\label{lem:matchingOfSize3}
Let $G=(V,E)$ be a 2VC graph without irrelevant edges and without non-isolating 2-vertex-cuts. Consider any partition $(V_1,V_2)$ of $V$ such that for each $i\in\{1,2\}$, $|V_i|\geq 3$ and if $|V_i|=3$, then $G[V_i]$ is a triangle. Then, there exists a matching of size $3$ between $V_1$ and $V_2$.
\end{restatable}
\begin{proof}
Consider the bipartite graph $F$ induced by the edges with exactly one endpoint in $V_1$. Let $M$ be a maximum (cardinality) matching of $F$. Assume by contradiction that $|M|\leq 2$. By K\"{o}nig-Egev\'{a}ry theorem\footnote{This theorem states that, in a bipartite graph, the cardinality of a maximum matching equals the cardinality of a minimum vertex cover, see e.g. \cite[Chapter~16]{S03}.} there exists a vertex cover $U$ of $F$ of size $|M|$. We distinguish $3$ subcases:

\begin{caseanalysis}
    \case{$U=\{u\}$.} Assume w.l.o.g. $u\in V_1$. Since $U$ is a vertex cover of $F$, there are no edges in $F$ (hence in $G$) between the non-empty sets $V_1\setminus \{u\}$ and $V_2$. Hence $U$ is a $1$-vertex-cut, a contradiction.
    
    \case{$U=\{u,v\}$, where $U$ is contained in one side of $F$.} Assume w.l.o.g. $U\subseteq V_1$. Since $|V_1|\geq 3$ and $U$ is a vertex cover of $F$, $U$ is a 2-vertex-cut separating $V_1\setminus U$ from $V_2$. This implies that $uv\notin E(G)$ (otherwise $uv$ would be an irrelevant edge), and in particular that $G[V_1]$ is not a triangle. Since $G[V_1]$ is not a triangle, $|V_1\setminus U|\geq 2$. This implies that $U$ is a non-isolating 2-vertex-cut of $G$, a contradiction.
    
    \case{$U=\{u,v\}$, where $u$ and $v$ belong to different sides of $F$.} Assume w.l.o.g. $u\in V_1$ and $v\in V_2$. Consider the sets $V'_1:=V_1\setminus \{u\}$ and $V'_2:=V_2\setminus \{v\}$, both of size at least $2$. Notice that there are no edges in $F$ (hence in $G$) between $V'_1$ and $V'_2$ (otherwise $U$ would not be a vertex cover of $F$). This implies that $U$ is a non-isolating 2-vertex-cut of $G$, a contradiction.
\end{caseanalysis}
\end{proof}

\subsection{Removing almost all small components}
\label{sec:removeSmall}

In this section we will prove Lemma~\ref{lem:smallComponents}, namely we will show how to get rid of small components other than pendant $4$-cycles. Recall that a connected component is large if it contains at least 6 edges and small otherwise. In our construction we will maintain that $S$ is canonical, hence in particular small components of the considered partial solution $S$ are cycles (of size $3$, $4$ or $5$).

We distinguish two types of small components that we deal with separately: small components that are adjacent to at least $2$ connected components of $S$, and small components adjacent only to a single connected component of $S$ (excluding pendant $4$-cycles). Let us start by handling the former type of small components. This is done via the following lemma.
\begin{lemma}\label{lem:nonLeafSmallComponent}
     Let $S$ be a canonical 2-edge-cover of a structured graph $G$. If $S$ contains a small component adjacent to at least two other connected components of $S$, then one can compute in polynomial time a canonical 2-edge-cover $S'$ of $G$, with strictly fewer components than $S$, and with $\cost(S')\leq \cost(S)$.
\end{lemma}
In order to prove the above lemma, we need the following definition and intermediate results.
 
\begin{definition}\label{def:shortcutPair}
     Let $C$ be a $k$-cycle, $k\in \{3,4,5\}$. We say that $\{u,v\}\subseteq V(C)$ is a \emph{shortcut pair} of $C$ if there is a matching $M$ of size $2$ from $\{u,v\}$ to $V(G)\setminus V(C)$, and there is a $u$-$v$ Hamiltonian path $P_{uv}$ in $G[V(C)]$ such that the first and last edge of $P_{uv}$ are in $E(C)$. We call $M$ the \emph{corresponding matching} of $\{u,v\}$ and $P_{uv}$ the \emph{shortcut path} of $\{u,v\}$ or the shortcut path of $C$ if $\{u,v\}$ is clear from the context.
\end{definition}

Shortcut pairs play an important role in the proof of Lemma~\ref{lem:nonLeafSmallComponent}. Intuitively, once we find a shortcut pair $\{u, v\}$ of a small component $C$ with corresponding matching $\{ux, vy\}$, we can exchange the edges of $C$ with its shortcut path (hence ``saving'' one edge), thus getting a path from $x$ to $y$ that includes all the nodes of $C$. We will use this to ``chain'' together several small components in order to build a path between components of $S$ with some desired properties. The next lemma shows that we can always find a shortcut pair of a small component. Additionally, it shows that we can choose such pair in a way that allows us to continue building the aforementioned path.

\begin{lemma}\label{lem:shortcutPair}
    Let $S$ be a canonical 2-edge-cover of a structured graph $G$, $C$ be a small component of $S$, and $wx$ be an edge of $G$ such that $w\in V(C), x\notin V(C)$. Then there exists a shortcut pair $\{u, v\}$ such that $ux$ is an edge of the corresponding matching of $\{u, v\}$. Furthermore, if $C$ is a 3-cycle, there exist two distinct shortcut pairs $\{u, v_1\}, \{u, v_2\}$ such that their corresponding matchings are $\{ux, v_1w_1\}$ and $\{ux, v_2w_2\}$, $w_1\neq w_2$, resp. All such pairs and their corresponding matchings and shortcut paths can be computed in polynomial time.
\end{lemma}
\begin{proof}
    By the 3-matching Lemma~\ref{lem:matchingOfSize3} applied to $C$, there is a matching $\{w_1x_1, w_2x_2, w_3x_3\}$, with $x_i\notin V(C), w_i\in V(C)$, for $i\in\{1, 2, 3\}$. W.l.o.g. we can assume $x_1=x$, indeed otherwise we can add $wx$ to the matching and remove $wx_i$ if any. If $C$ is a 3-cycle, the pairs $\{u, v_1\}=\{w_1, w_2\}$ and $\{u, v_2\}=\{w_1, w_3\}$ are the desired pairs, with shortcut paths $P_{uv_1}=C\setminus\{uv_1\}$ and $P_{uv_2}=C\setminus\{uv_2\}$, resp. 
    
    Assume next that $C$ is a 4-cycle or a 5-cycle. If $w_i$ is adjacent to $w_1$ in $C$ for some $i\in\{2, 3\}$ then $\{u, v\}=\{w_1, w_i\}$ is the desired shortcut pair, with corresponding matching $\{ux, vx_i\}$ and shortcut path $P_{uv}=C\setminus\{uv\}$. Notice that this is always the case if $C$ is a $4$-cycle. We next assume that $w_i, i\in\{2, 3\}$, is not adjacent to $w_1$ in $C$ (and thus $C$ is a $5$-cycle). Moreover, we can assume that there is no matching $M = \{w_1'x, w_2'x_2'\}$, where $x_2'\notin V(C), w_1', w_2'\in V(C)$, and $w_1'$ is adjacent to $w_2'$ in $C$.
    
    Let $C=w_1aw_2w_3b$. Notice that by the above assumption there cannot exist an edge $ay\in E(G)$ with $y\notin V(C)$. Indeed, if $y\neq x$, then the matching $M = \{w_1x, ay\}$ is a contradiction to our above assumption, and if $y=x$ then the matching $M = \{ax, w_2x_2\}$ is a contradiction to our above assumption. Symmetrically, $b$ is not adjacent to any node outside $C$ in $G$.

    Since $G$ is a structured graph, $C$ is not a removable $5$-cycle, and thus either $a$ or $b$, say $a$ by symmetry, is adjacent in $G$ to some node in $C$ other than $w_1$ and $w_2$. We claim that $\{u, v\}=\{w_1, w_2\}$ is our desired shortcut pair, with its corresponding matching $\{ux, vx_2\}$. If $ab\in G$ then $P_{uv}=w_1abw_3w_2$ satisfies the claim. If $ab\notin E(G)$, then it must be $w_3a\in E(G)$. Hence  $P_{uv}=w_1bw_3aw_2$ satisfies the claim.
\end{proof}
\begin{corollary}\label{cor:shortcutPair}
    Let $S$ be a canonical 2-edge-cover of $G$ and $C$ be a small component of $S$. If $C$ is adjacent to at least two other connected components of $S$, then we can find in polynomial time a shortcut pair $\{u, v\}$ of $C$ such that its corresponding matching $\{ux_1, vx_2\}$ has $x_1\in V(C_1), x_2\in V(C_2)$, where $C$, $C_1$ and $C_2$ are distinct connected components of $S$.
\end{corollary}
\begin{proof}
    We first claim that there are edges $ax_1, bx_2$ with $a, b\in V(C), x_1\in V(C_1), x_2\in V(C_2), a\neq b$, where $C_1$ and $C_2$ are distinct connected components of $S$. To obtain this, we apply the 3-matching Lemma~\ref{lem:matchingOfSize3} to $C$. If two edges of the matching have their endpoints not in $C$ in distinct connected components of $S$ we are done. Otherwise, all edges of the matching are incident to the same connected component $C'$. By assumption of the lemma there is at least one edge between $C$ and a connected component of $S$ distinct from $C'$, so we can modify the matching by including that edge and removing at most one edge. The claim follows. 
    
    Apply Lemma~\ref{lem:shortcutPair} to $ax_1$ and $bx_2$ to find shortcut pairs of $C$, $\{u_1, v_1\}, \{u_2, v_2\}$, such that $u_1x_1, u_2x_2$ belong to their respective corresponding matchings. If the corresponding matching $\{u_1x_1, v_1y_1\}$ has $y_1\notin V(C_1)$, then $\{u, v\}=\{u_1, v_1\}$ is the desired shortcut pair. Similarly for the corresponding matching to $\{u_2, v_2\}$, $\{u_2x_2, v_2y_2\}$. Assume now that $y_1\in V(C_1), y_2\in V(C_2)$.
    
    If $u_1=u_2$, then $\{u, v\}=\{u_1, v_2\}$ is the desired shortcut pair. Similarly with $u_1=v_2, v_1=u_2, v_1=v_2$. Assume $u_1, v_1, u_2, v_2$ are all distinct. Then either $u_1$ is adjacent in $C$ to one of $u_2$ or $v_2$, or $v_1$ is. Say $u_1$ is adjacent to $u_2$. The shortcut pair $\{u, v\}=\{u_1, u_2\}$ satisfies the claim with matching $\{u_1x_1, u_2x_2\}$ and $P_{uv}=C\setminus\{uv\}$. The other cases are treated identically, with the corresponding matchings being $\{u_1x_1, v_2y_2\}$, $\{v_1y_1, u_2x_2\}$, and $\{v_1y_1, v_2y_2\}$.
\end{proof}
\SetKwComment{Comment}{$\triangleright$\ }{}\setcounter{AlgoLine}{0}
\begin{algorithm}
\caption{Algorithm that computes $P$. During the execution of the algorithm, at the beginning of every iteration of a while loop, we maintain the notation that $P = u_Lu_1\dots u_ku_R$, $u_L\in V(C_L)$, $u_R\in V(C_R)$.}
\label{alg:pathP}
\KwIn{A canonical $2$-edge-cover $S$ of a structured graph $G$ containing a small component $C$ adjacent to at least two other components of $S$}
\KwOut{A path $P$ between connected components of $S$ satisfying Properties~\ref{prp:pathLength} to~\ref{prp:u2}.}
\textbf{Initialization:} Let $\{u, v\}$ be a shortcut pair of $C$ with shortcut path $P_{uv}$ and corresponding matching $\{uu_L, vu_R\}$, $u_L\in V(C_L)$, $u_R\in V(C_R)$, where $C_L$ and $C_R$ are distinct connected components of $S$\; \label{line:initP}
$P \gets P_{uv} \cup \{uu_L, vu_R\}$\; \label{line:initP2}

\While{$C_L$ is small}{
    \If{there exists a shortcut pair $\{u, v\}$ of $C_L$ with shortcut path $P_{uv}$ and corresponding matching $\{uu_1, vx\}$ such that $x \notin V(C_R) \cup V(P)$\label{line:ifCL}}{
        Let $C_L'$ be the component of $S$ such that $x \in V(C_L')$\;
        $P \gets (E(P) \setminus \{u_1u_L\}) \cup E(P_{uv}) \cup \{uu_1, vx\}$\;
        $C_L \gets C_L'$ \label{line:updatePCL} \Comment*[r]{Extension on the left side (see Figure~\ref{fig:extendingPath})}
    }
    \Else{
        Let $\{u, v\}$ be a shortcut pair of $C_L$ with corresponding matching $\{uu_1, vx\}$ such that $x \in V(C_R) \cup V(P)$. If $C_L$ is a $3$-cycle, choose such a pair such that $x \neq u_3$\; \label{line:triangleCL}
        $P \gets (E(P) \setminus \{u_1u_L\}) \cup \{uu_1\}$ \Comment*[r]{Stop the extension on the left side (see Figure~\ref{fig:closingPath})}
        \textbf{break} \Comment*[r]{We exit the while loop}
    }
}

\While{$C_R$ is small}{
    \If{there exists a shortcut pair $\{u, v\}$ of $C_R$ with shortcut path $P_{uv}$ and corresponding matching $\{uu_k, vx\}$ such that $x \notin V(C_L) \cup V(P)$\label{line:ifCR}}{
        Let $C_R'$ be the component of $S$ such that $x \in V(C_R')$\;
        $P \gets (E(P) \setminus \{u_ku_R\}) \cup E(P_{uv}) \cup \{uu_k, vx\}$\;
        $C_R \gets C_R'$ \label{line:updatePCR} \Comment*[r]{Extension on the right side}
    }
    \Else{
        Let $\{u, v\}$ be a shortcut pair of $C_R$ with corresponding matching $\{uu_k, vx\}$ such that $x \in V(C_L) \cup V(P)$. If $C_R$ is a $3$-cycle, choose such a pair such that $x \neq u_{k-2}$\; \label{line:triangleCR}
        $P \gets (E(P) \setminus \{u_ku_R\}) \cup \{uu_k\}$ \Comment*[r]{Stop the extension on the right side}
        \textbf{break} \Comment*[r]{We exit the while loop}
    }
}

\Return{$P$}\;

\end{algorithm}
We are now ready to prove Lemma \ref{lem:nonLeafSmallComponent}.
\begin{proof}[Proof of Lemma \ref{lem:nonLeafSmallComponent}]
    We next describe an iterative procedure to construct a (maximal in some sense) path $P$, incident on two distinct connected components $C_L$ and $C_R$ of $S$. Formally, Algorithm~\ref{alg:pathP} finds, in polynomial time, a path $P = u_Lu_1u_2\dots u_ku_R$ satisfying the following properties:

    \begin{enumerate}
        \item $|E(P)|\geq 4$.\label{prp:pathLength}
        \item The set $\{u_1, \dots, u_k\}$ of internal nodes of $P$ is the union of the node sets of some small components of $S$. The nodes $u_L, u_R$ belong to distinct connected components $C_L, C_R$ of $S$.\label{prp:internalNodesP}
        \item If $C_L$ is small, then for some $v_L\in V(C_L)$, $\{u_L, v_L\}$ is a shortcut pair of $C_L$ such that its corresponding matching is $\{u_Lu_1, v_Lw_L\}$, with $w_L\in V(C_R)\cup V(P)$. Moreover, if $C_L$ is a $3$-cycle, then $w_L\neq u_3$.\label{prp:cl}
        \item If $C_R$ is small, then for some $v_R\in V(C_R)$, $\{u_R, v_R\}$ is a shortcut pair of $C_R$ such that its corresponding matching is $\{u_Ru_1, v_Rw_R\}$, with $w_R\in V(C_L)\cup V(P)$. Moreover, if $C_R$ is a $3$-cycle, then $w_R\neq u_{k-2}$.\label{prp:cr}
        \item The edges $u_1u_2$ and $u_{k-1}u_k$ are edges of small components of $S$.\label{prp:u2}
    \end{enumerate}

    Since at every iteration the length of $P$ increases or a while loop stops, Algorithm~\ref{alg:pathP}, terminates after a polynomial number of iterations. By Corollary~\ref{cor:shortcutPair}, one can find the shortcut pair required in Line~\ref{line:initP} in polynomial time. The condition in Line~\ref{line:ifCL} can be checked in polynomial time by considering all pairs of nodes of $C_L$ (recall $C_L$ is small when executing Line~\ref{line:ifCL}) and all matchings of size $2$ incident to such pairs. When executing Line~\ref{line:triangleCL}, by Lemma~\ref{lem:shortcutPair} and the fact that the condition at Line~\ref{line:ifCL} does not hold, we can find the desired pair. Notice that here we used the fact that, if $C_L$ is a $3$-cycle, Lemma~\ref{lem:shortcutPair} guarantees the existence of two shortcut pairs $\{u, v_1\}, \{u, v_2\}$ with corresponding matchings $\{uu_1, v_1w_1\}$ and $\{uu_1, v_2w_2\}, w_1\neq w_2$. Since the condition at Line~\ref{line:ifCL} does not hold, $w_1, w_2\in V(C_R)\cup V(P)$, so we can choose $v = v_i$ such that $w_i \neq u_3$, for some $i\in\{1, 2\}$. The executions of Lines~\ref{line:ifCR} and~\ref{line:triangleCR} can also be done in polynomial time by identic arguments. Therefore, Algorithm~\ref{alg:pathP} runs in polynomial time.

    Now we show that $P$ satisfies Properties~\ref{prp:pathLength} to~\ref{prp:u2}. Property~\ref{prp:pathLength} is satisfied when initializing $P$ at Line~\ref{line:initP2}, and is maintained through the execution of the algorithm because we never decrease the size of $P$. Notice that $P$ satisfies Property~\ref{prp:internalNodesP} at Line~\ref{line:initP2}. The set of internal nodes of $P$ then only changes at Lines~\ref{line:updatePCL} and~\ref{line:updatePCR}. When executing those lines, we add the nodes of the shortcut path $P_{uv}$ of $C_L$ (resp., $C_R$) to the set of internal nodes of $P$. Since when executing Line~\ref{line:updatePCL} (resp.,~\ref{line:updatePCR}) $C_L$ (resp., $C_R$) is small, we maintain Property~\ref{prp:internalNodesP}. Properties~\ref{prp:cl} and~\ref{prp:cr} are satisfied by the choice of $\{u, v\}$ at Lines~\ref{line:triangleCL} and~\ref{line:triangleCR}, respectively. Finally, notice that the edges $u_1u_2$ and $u_{k-1}u_k$ are the last edges of some shortcut path of some small component. Thus, by Definition~\ref{def:shortcutPair}, they are edges of small components of $C$. Therefore, $P$ satisfies Properties~\ref{prp:pathLength} to~\ref{prp:u2}, as we wanted to prove.
    
    Next we build $S'$ as follows. Let $C_1,\ldots,C_q$ be the small components whose node sets form the internal nodes of $P$ (Property~\ref{prp:internalNodesP}). We initially set $S':=(S\setminus \bigcup_{i=1}^{q}E(C_i)) \cup E(P)$. If $C_L$ is small, let $\{u_L,v_L\}$ be the shortcut pair given by Property~\ref{prp:cl}, with matching $\{u_Lu_1, v_Lw_L\}$ and shortcut path $P_{u_Lv_L}$: replace $S'$ with $(S'\setminus E(C_L))\cup P_{u_Lv_L}\cup \{v_Lw_L\}$. If $C_R$ is small, we update $S'$ symmetrically (using Property~\ref{prp:cr}). We will use figures to show examples of the construction of $S'$ based on $P$ in the corresponding relevant cases.
       
    Notice that $|S'|=|S|+1$. We also remark that we collect $(|E(P)|-1)\cdot \frac{1}{3}$ credits from the removed small components. Clearly $S'$ has fewer connected components than $S$ and is canonical. We will show that in most cases $\cost(S')\leq \cost(S)$. When the latter property does not hold, we will describe an alternative $S'$ that satisfies the claim. We consider different cases depending on the type of the components $C_R$ and $C_L$. 
    
    \begin{figure}
        \centering
        \includegraphics{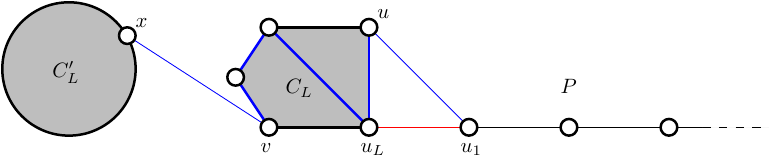}
        \caption{An example of extending $P$ on the $C_L$ side. The shortcut path $P_{uv}$ of $C_L$ is shown in bold blue. In this iteration we remove $u_Lu_1$ (in red) from $P$, and add $P_{uv}, u_1u$ and $vx$ (all of them in blue) to $P$. Here $C_L'$ could be small or large. In the next iteration $C_L'$ takes the place of $C_L$.}
        \label{fig:extendingPath}
    \end{figure}
    
    \begin{cfigure}
        \centering
        \includegraphics{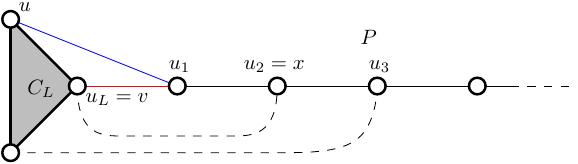}
        \caption{In this case we stop extending $P$ on the left side. Notice that in this case we have at least $2$ choices for the pair $\{u, v\}$. Since $C_L$ is a $3$-cycle we select the pair $\{u, v\}$ such that $x\neq u_3$. The dashed edges are edges in $G$.}
        \label{fig:closingPath}
    \end{cfigure}

    \begin{caseanalysis}
    \case{Both $C_R$ and $C_L$ are large.}\hfill\\
    An example of the construction of $S'$ for this case is shown in Figure~\ref{fig:closingPathBothLarge}. $S'$ contains a (complex) connected component $C$ spanning the nodes of $C_R$, $C_L$ and $P$. We need $1$ credit for $C$ and $|E(P)|\cdot \frac{1}{4}$ credits for the edges of $P$ (which become bridges of $C$). Recall that we collect $(|E(P)-1|)\cdot \frac{1}{3}$ credits from the removed small components. We also collect $1$ credit from $C_R$ and $C_L$ each, while their blocks retain their credits. Altogether
    \begin{align*}
    \cost(S) - \cost(S') & = |S|-|S'|+\credit(C_R)+\credit(C_L) + (|E(P)|-1)\cdot 1/3 - \credit(C) - |E(P)|\cdot 1/4 \\
    & = -1+1+1+(|E(P)|-1)\cdot 1/3 -1 - |E(P)|\cdot 1/4 \geq 0,
    \end{align*}
    where in the last inequality we used $|E(P)|\geq 4$ (Property~\ref{prp:pathLength}).

    \begin{figure}
    \centering
    \includegraphics{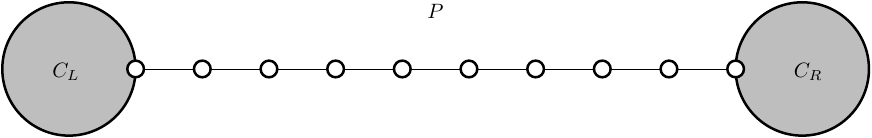}
    \caption{An example of the construction of $S'$ when both $C_L$ and $C_R$ are large components. The set of internal nodes of $P$ is the union of the node sets of some small components $C_1,\dots, C_q$ of $S$. In this case we simply set $S':=(S\setminus \bigcup_{i=1}^{q}E(C_i)) \cup E(P)$. Every connected component $C\notin\{C_L, C_R\}\cup\{C_1,\dots, C_q\}$ of $S$ remains unaltered.}
    \label{fig:closingPathBothLarge}
    \end{figure}
    
    \case{At least one among $C_R$ and $C_L$ is small  (w.l.o.g. assume $C_L$ is small).}\hfill\\
    Notice that, by Property~\ref{prp:cl}, one has that $w_L\in V(C_R)\cup V(P)\setminus\{u_1\}$, so that either $w_L\in V(C_R)$ or $w_L = u_j$ for some $j\in\{2, \dots, k\}$.
    
    \subcase{$w_L\in V(C_R)$.}\hfill\\
    An example of the construction of $S'$ for this case is shown in Figure~\ref{fig:closingPathInCR}. Regardless of the type of component $C_R$ is, the nodes of $C_R, C_L$ and $P$ are merged into a unique large component $C$ of $S'$, and the nodes of $C_L$ and $P$ belong to a block $B$ in $S'$. Therefore the edges of $P$ are not bridges in $C$ and hence have no credits in $S'$.

    \begin{figure}
        \centering
        \includegraphics{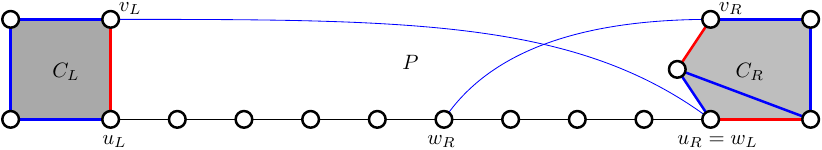}
        \caption{An example of the construction of $S'$ when $C_L$ is a small component and $w_L\in V(C_R)$. We show the case in which $C_R$ is also a small component. We remove the edges of $C_L$ and $C_R$, and we add their shortcut paths and the edges $v_Lw_L$ and $v_Rw_R$. This is equivalent to removing the red edges and adding the blue edges. The shortcut paths of $C_L$ and $C_R$ are shown in bold blue. In this case every node of $P$ and $C_L$ belongs to a block $B$ in $S'$. If $C_R$ is small, the nodes of $C_R$ also belong to $B$. The set of internal nodes of $P$ is the union of the node sets of some small components $C_1,\dots, C_q$ of $S$. Every connected component $C\notin\{C_L, C_R\}\cup\{C_1,\dots, C_q\}$ of $S$ remains unaltered.}
        \label{fig:closingPathInCR}
    \end{figure}
    
    If $C_R$ is also a small component the nodes of $C_R$ (and thus all nodes of $C$) also belong to $B$, so $B$ is the only new block. To see this, notice that all the nodes of the path from $w_R$ to $u_R$ in $S'$ belong to a block $B'$ together with the nodes of $C_R$. Since $w_R\in V(C_L)\cup V(P)$ (Property~\ref{prp:cr}), this implies that $B$ and $B'$ share at least one edge, and thus it must be that $B=B'$. If $C_R$ is a large component, then $B$ is also the only block of $S'$ not present in $S$. Using $|E(P)|\geq 4$ (Property~\ref{prp:pathLength}) one has
    \begin{align*}
    \cost(S) - \cost(S') & = |S|-|S'|+\credit(C_R)+\credit(C_L) + (|E(P)|-1)\cdot 1/3 - \credit(C) -\credit(B)\\
    & \geq -1 + 1 + 1 + (|E(P)|-1)\cdot 1/3 - 1 -1\geq 0.
    \end{align*}
    
    \subcase{$w_L=u_2$.}\label{case:wlu2}\hfill\\
    We illustrate this case in Figure~\ref{fig:caseb2}. Let $C$ be the component of $S$ such that $w_L\in C$. By Property~\ref{prp:u2}, $u_1u_2\in E(C)$. Let $S'':= (S\setminus E(C_L))\cup E(P_{u_Lv_L})\cup\{u_Lu_1, v_Lu_2\}\setminus \{u_1u_2\}$. In $S''$, the nodes of $C$ and $C_L$ belong to a single cycle $C'$ of size at least $6$: $C'$ is a 2VC large component, and thus it contains exactly one block $B$. Then, $S''$ is a solution with fewer components than $S$ and one has
    \begin{align*}
        \cost(S) - \cost(S'') = \credit(C_L) + \credit(C) - \credit(C') -\credit(B)\geq 1 + 1 - 1 - 1 = 0.
    \end{align*}
        
    \begin{figure}
        \centering
        \includegraphics{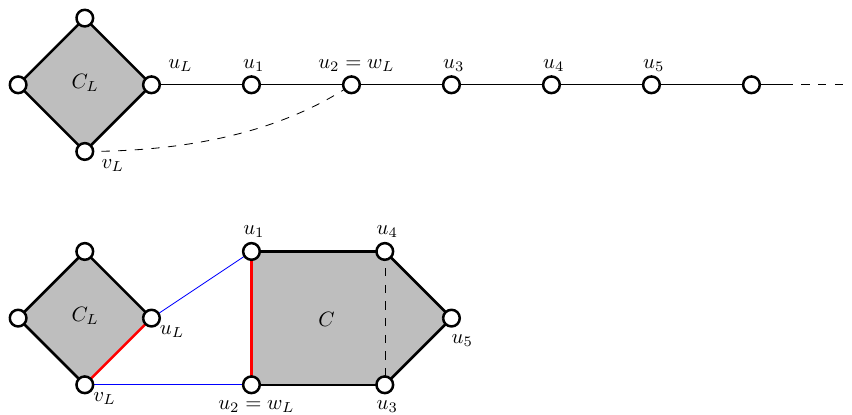}
        \caption{Lemma~\ref{lem:nonLeafSmallComponent} Case~\ref{case:wlu2}. By Property~\ref{prp:u2}, the edge $u_1u_2$ belongs to some small component $C$ of $S$. Note that $u_3u_4\in E(P)$. We can add the blue edges and remove the red edges from $S$ to merge two small components into a large one.}
        \label{fig:caseb2}
    \end{figure}
    
    \subcase{$w_L = u_j, j\in\{3,\dots, k\}$.}\hfill\\
    \subsubcase{$C_R$ is a large component.}\hfill\\
    An example of the construction of $S'$ for this case is shown in Figure~\ref{fig:closingPathCRLarge}. $S'$ contains a complex connected component $C$ spanning the nodes of $C_R, C_L$ and $P$. The nodes of $C_L$ together with the path $u_Lu_1u_2\dots u_j$ and the edge $v_Lu_j$ belong now to a block $B_L$ in $S'$, so exactly $|E(P)-j|$ edges of $P$ are now bridges of $C$, and require $\frac{1}{4}$ credit each. If $C_L$ is not a $3$-cycle one has
       \begin{align*}
           &\cost(S) - \cost(S') \geq \\
           &|S| - |S'| + \credit(C_R) + \credit(C_L) + (|E(P)|-1)\cdot 1/3 - \credit(C) - \credit(B_L) - (|E(P)|-j)\cdot 1/4\geq\\
           &-1 + 1 + 4/3 + (|E(P)|-1)\cdot 1/3 - 1 - 1 - (|E(P)|-3)\cdot 1/4 > 0.
       \end{align*}

       \begin{cfigure}
        \centering
        \includegraphics{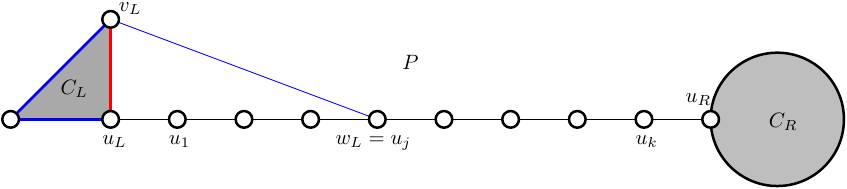}
        \caption{An example of the construction of $S'$ when $C_L$ is a small component, $w_L = u_j, j\in\{3, \dots, k\}$ and $C_R$ is a large component. We remove the edges of $C_L$ and we add the shortcut path of $C_L$ and the edge $v_Lw_L$. This is equivalent to removing the red edge and adding the blue edges. The shortcut path of $C_L$ is shown in bold blue. The nodes of $C_L$ together with the path $u_Lu_1u_2\dots u_j$ and the edge $v_Lu_j$ belong to a block $B_L$ in $S'$. The set of internal nodes of $P$ is the union of the node sets of some small components $C_1,\dots, C_q$ of $S$. Every connected component $C\notin\{C_L, C_R\}\cup\{C_1,\dots, C_q\}$ of $S$ remains unaltered.}
        \label{fig:closingPathCRLarge}
    \end{cfigure}
        On the other hand, if $C_L$ is a $3$-cycle by Property~\ref{prp:cl} we can assume $j\geq 4$, so one has
       \begin{align*}
           &\cost(S) - \cost(S') \geq \\
           &|S| - |S'| + \credit(C_R) + \credit(C_L) + (|E(P)|-1)\cdot 1/3 - \credit(C) - \credit(B_L) - (|E(P)|-j)\cdot 1/4\geq\\
           &-1 + 1 + 1 + (|E(P)|-1)\cdot 1/3 - 1 - 1 - (|E(P)|-4)\cdot 1/4 \geq 0.
       \end{align*}
    
    In both cases we used Property~\ref{prp:pathLength}, i.e., $|E(P)|\geq 4$.
    
    \subsubcase{$C_R$ is a small component.}\hfill\\
    An example of the construction of $S'$ for this case is shown in Figure~\ref{fig:closingPathCRSmall}. $S'$ contains a large component $C$ spanning the nodes of $C_R$, $C_L$ and $P$. If in $S'$ the nodes of $C_R, C_L$ and $P$ belong to a single block $B$ then one has 
    \begin{align*}
     \cost(S) - \cost(S') &= |S| - |S'| + \credit(C_R) + \credit(C_L) + (|E(P)|-1)\cdot 1/3 - \credit(C) -\credit(B)\\
     &\geq -1 + 1 + 1 + 1 - 1 - 1 = 0,
    \end{align*}
    where we have used $|E(P)|\geq 4$ (Property~\ref{prp:pathLength}). Otherwise, $C$ is a complex component. As in the previous case, the nodes of $C_L$ together with the path $u_Lu_1u_2\dots u_j$ and the edge $v_Lu_j$ belong now to a block $B_L$ in $S'$. By the symmetry of Properties~\ref{prp:pathLength} to~\ref{prp:u2} w.r.t. $C_L$ and $C_R$, we can assume that the nodes of $C_R$ together with the subpath of $P$, $u_Ru_ku_{k-1}\dots u_{k-j'}$ and an edge $v_Ru_{k-j'}$ belong now to a block $B_R$, where $j'\geq 3$ and $k-j'\geq j$ (otherwise $B_R$ and $B_L$ would share an edge in $S'$ and thus $C$ would not be a complex component). Notice that we have $|E(P)|\geq j+j'$. If $C_L$ is not a $3$-cycle one has
    \begin{align*}
        \cost(S) - \cost(S') &= |S| - |S'| + \credit(C_R) + \credit(C_L) + (|E(P)|-1)\cdot 1/3 \\ &-\credit(C) - \credit(B_R) - \credit(B_L) - (|E(P)|-j-j')\cdot 1/4 \\ &\geq
        -1 + 1 + 4/3 + (|E(P)|-1)\cdot 1/3 - 1 - 1 - 1 - (|E(P)|-6)\cdot 1/4 \geq 0.
    \end{align*}
    
    The last inequality above follows from the fact that $|E(P)|\geq j+j'\geq 6$. If $C_L$ is a $3$-cycle we can assume by Property~\ref{prp:cl} that $j\geq 4$ and one has 
    \begin{align*}
        \cost(S) - \cost(S') &= |S| - |S'| + \credit(C_R) + \credit(C_L) + (|E(P)|-1)\cdot 1/3 \\ &- \credit(C) - \credit(B_R) - \credit(B_L) - (|E(P)|-j-j')\cdot 1/4 \\ &\geq
        -1 + 1 + 1 + (|E(P)|-1)\cdot 1/3 - 1 - 1 - 1 - (|E(P)|-7)\cdot 1/4 \geq 0.
    \end{align*}
    \end{caseanalysis}
    The last inequality above follows from the fact that $|E(P)|\geq j+j'\geq 7$.
    \begin{figure}
        \centering
        \includegraphics{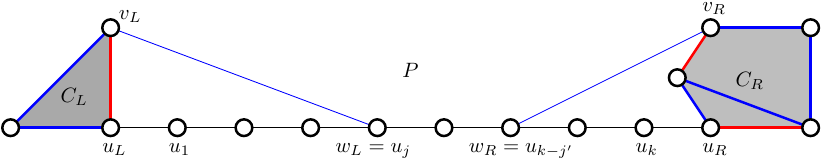}
        \caption{An example of the construction of $S'$ when $C_L$ is a small component, $w_L = u_j, j\in\{3, \dots, k\}$ and $C_R$ is also a small component. We remove the edges of $C_L$ and $C_R$, and we add their shortcut paths and the edges $v_Lw_L$ and $v_Rw_R$. This is equivalent to removing the red edges and adding the blue edges. The shortcut paths of $C_L$ and $C_R$ are shown in bold blue. The nodes of $C_L$ together with the path $u_Lu_1u_2\dots u_j$ and the edge $v_Lw_L$ belong to a block $B_L$ in $S'$. Similarly, the nodes of $C_R$ together with the path $u_Ru_ku_{k-1}\dots u_{k-j'}$ and the edge $v_Rw_R$ belong to a block $B_R$ in $S'$. In the case shown, $B_L\neq B_R$. The set of internal nodes of $P$ is the union of the node sets of some small components $C_1,\dots, C_q$ of $S$. Every connected component $C\notin\{C_L, C_R\}\cup\{C_1,\dots, C_q\}$ of $S$ remains unaltered.}
        \label{fig:closingPathCRSmall}
    \end{figure}
\end{proof}

We next address the small components of the second type, i.e., small components adjacent only to a single connected component of $S$ (excluding pendant $4$-cycles).

\begin{lemma}\label{lem:leafSmallComponent}
    Let $S$ be a canonical 2-edge-cover of a structured graph $G$. Assume that $S$ contains a small component $C$ such that $C$ is only adjacent to another component $C'$ of $S$ and $C$ is not a pendant $4$-cycle. Then one can compute in polynomial time a canonical 2-edge-cover $S'$ of $G$ with strictly fewer components than $S$ and with $\cost(S')\leq \cost(S)$.
\end{lemma}
\begin{proof}
    Apply the 3-matching Lemma~\ref{lem:matchingOfSize3} to $C$ to get a matching $\{u_1v_1, u_2v_2, u_3v_3\}, u_i\in C, v_i\in C'$ for $i\in\{1,2,3\}$. W.l.o.g. assume that $u_1$ and $u_2$ are adjacent in $C$. This way, $E(C)\cup\{u_1v_1, u_2v_2\}\setminus\{u_1u_2\}$ is a path from $v_1$ to $v_2$ whose internal nodes are given by $V(C)$. We will consider a few cases depending on the type of $C$ and $C'$.

\begin{caseanalysis}
\case{$C'$ is small.}\hfill\\
\subcase{At least one among $C$ and $C'$ is a $3$-cycle or both $C$ and $C'$ are $4$-cycles}.\label{case:both4Cycles}\hfill\\
In this case w.l.o.g. we can assume that $v_1$ and $v_2$ are adjacent in $C'$. Let $S':=S\cup\{u_1v_1, u_2v_2\}\setminus\{u_1u_2, v_1v_2\}$. Notice that $E(C)\cup E(C')\cup\{u_1v_1, u_2v_2\}\setminus\{u_1u_2, v_1v_2\}$ induces a cycle $C''$ of length at least 6 in $S'$ (hence a 2VC large component with a single block $B$). One has 
    $$
    \cost(S) - \cost(S') = \credit(C) + \credit(C') - \credit(C'') - \credit(B)\geq 1 + 1 - 1 - 1 = 0.
    $$

\subcase{At least one among $C$ and $C'$ is a $5$-cycle.}\hfill\\
Assume that $C'$ is a $5$-cycle, the other case being symmetric. By Case~\ref{case:both4Cycles} we can assume that $C$ is a $4$-cycle or a $5$-cycle. Set $S':=S\cup \{u_1v_1, u_2v_2\}\setminus \{u_1u_2\}$. Notice that $C$ and $C'$ are merged into a 2VC large component $C''$ of $S'$ with a single block $B$. One has 
    $$
    \cost(S) - \cost(S') = \credit(C) + \credit(C') - \credit(C'') -\credit(B) - 1\geq 4/3 + 5/3 - 1 - 1 - 1 = 0.
    $$ 
\case{$C'$ is large.}\hfill\\
Observe that $C$ is a $3$-cycle or a $5$-cycle since by assumption it is not a pendant $4$-cycle.

\subcase{$C$ is a $5$-cycle.}\hfill\\
\subsubcase{There is a $v_1$-$v_2$ path $P$ in $C'$ that includes some edge of some block or at least $2$ bridges.}\label{case:v1v2block}\hfill\\
Set $S':=S\cup \{u_1v_1, u_2v_2\}\setminus \{u_1u_2\}$. Let $C''$ be the connected component of $S'$ that includes the nodes of $C$ and $P$. If $P$ includes an edge of some block $B$ in $S$, then in $S'$ the nodes of $C, P$ and $B$ belong to a single block $B'$, and one has:
    $$
    \cost(S)-\cost(S')\geq |S| - |S'| + \credit(C)+ \credit(B) - \credit(B') = - 1 + \frac{5}{3} + 1 - 1 > 0.
    $$

Otherwise, $P$ includes at least $2$ bridges, so we get at least $2\cdot\frac{1}{4}=\frac{1}{2}$ credits from the bridges along $P$. Also, the nodes of $C$ and $P$ form a new block $B'$, which is the only block in $S'$ not present in $S$. Altogether
    $$
    \cost(S)-\cost(S')\geq |S| - |S'| + \credit(C)+\frac{1}{2} - \credit(B') = - 1 + \frac{5}{3} + \frac{1}{2} - 1 > 0.
    $$
    
\subsubcase{$v_1v_2$ is a bridge of $C'$.}\label{case:v1v2bridge}\hfill\\
Set $S':=S\cup \{u_1v_1, u_2v_2\}\setminus \{u_1u_2, v_1v_2\}$. In $S'$ we replace the bridge $v_1v_2$ with  a path $P$ of $|E(C)|+1$ bridges. One has 
    $$
    \cost(S)-\cost(S')= \credit(C)+\credit(v_1v_2)-\sum_{e\in E(P)}\credit(e)= |E(C)|\frac{1}{3} + \frac{1}{4} - (|E(C)|+1)\frac{1}{4} \geq 0.
    $$ 
\subcase{$C$ is a $3$-cycle.}\hfill\\
\subsubcase{There is a $v_i$-$v_j$ path $P$ in $C'$, for some $i,j$, that includes some edge of some block or at least $4$ bridges.}\hfill\\
W.l.o.g. assume $(i,j)=(1,2)$. In this case we build $S'$ as in Case~\ref{case:v1v2block}. The cost analysis works almost identically as in Case~\ref{case:v1v2block}. In this case $C$ provides $1$ credit only (rather than $5/3$), but in both subcases this is enough. In particular, in the case $P$ includes at least $4$ bridges they provide $4\cdot \frac{1}{4} = 1$ credits which makes up for the loss of credit of $C$.

\subsubcase{$v_iv_j$ is a bridge of $C'$, for some $i,j$.}\hfill\\
W.l.o.g. assume $(i,j)=(1,2)$. In this case we build $S'$ like in Case~\ref{case:v1v2bridge}. The analysis is identical.

\subsubcase{None of the above cases.}\label{case:triangleLast}\hfill\\
This case is illustrated in Figure \ref{fig:c2iii}. Let $P_{v_iv_j}$ be the $v_i$-$v_j$ path in $C'$. Observe that such paths consist of bridges only, and have length $2$ or $3$. W.l.o.g. assume $|E(P_{v_1v_2})|\ge |E(P_{v_1v_3})|,|E(P_{v_2v_3})|$. Since $2\leq |E(P_{v_1v_2})|\leq 3$, $v_3$ cannot be an internal vertex of $P_{v_1v_2}$: indeed otherwise the path from $v_3$ to at least one of $v_1$ or $v_2$ would have length $1$, a contradiction. However since $P_{v_1v_2}$ consists only of bridges and $|E(P_{v_1v_2})|\ge |E(P_{v_1v_3})|,|E(P_{v_2v_3})|$, at least one internal node $w$ of $P_{v_1v_2}$ has degree at least $3$ in $C'$. Note $w$ is adjacent to $v_1$ or $v_2$ in $C'$: assume w.l.o.g. $v_1w\in E(C')$. Set $S':=S\cup\{u_1v_1, u_2v_2\}\setminus\{u_1u_2, v_1w\}$. The nodes of $C$ and $C'$ belong to a single connected component $C''$ in $S'$. In $C''$  the bridges and blocks of $C'$ remain the same, except for the removed bridge $v_1w$. The edges of the path $P=v_1u_1u_3u_2v_2$ are bridges of $C''$. Altogether 
    $$
    \cost(S)-\cost(S')\geq \credit(C)+\credit(v_1w)+\credit(C')-\credit(C'')-\sum_{e\in E(P)}\credit(e)=1+\frac14-4\cdot \frac14 > 0.
    $$ 

\begin{figure}
    \centering
    \includegraphics[scale=0.8]{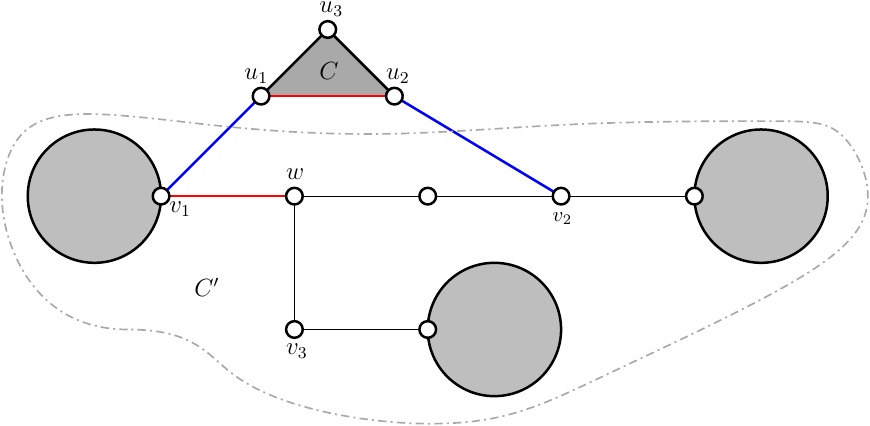}
    \caption{Illustration of Case~\ref{case:triangleLast} of Lemma \ref{lem:leafSmallComponent}. $S'$ is obtained from $S$ by adding the blue edges and removing the red ones.} 
    \label{fig:c2iii}
\end{figure}
\end{caseanalysis}

In each case it is easy to check that $S'$ is a 2-edge-cover and the number of components of $S'$ is strictly less than the number of components of $S$. Furthermore, $S'$ is canonical since each newly created block has at least $5$ nodes and each newly created connected component has at least $6$ edges (hence it is large).
\end{proof}

Lemma~\ref{lem:smallComponents} follows by repeated applications of Lemma~\ref{lem:nonLeafSmallComponent} and Lemma~\ref{lem:leafSmallComponent}. 
 
\subsection{Removing complex components}
\label{sec:removeComplex}

In this section we prove Lemma~\ref{lem:complexComponents}. Recall that initially we are given a canonical 2-edge-cover $S$ where the only small components are pendant $4$-cycles: we will maintain this invariant.

\begin{definition}\label{def:bridgePath}
    We say that $P$ is an \emph{extending path} of a complex component $C$, if it is a path in $E(G)\setminus E(C)$ between two nodes of $C$ such that $P=\{e_0\}\cup P_1\cup\{e_1\}\cup P_2\cup \{ e_2\}\dots\cup P_k\cup\{e_k\}$, $k\geq 0$, where $P_1,\dots,P_k$ are paths (possibly of length zero) in distinct components $C_1,\dots,C_k$, resp., of $S\setminus E(C)$, and $e_0,\dots, e_k$ are edges in $E(G)\setminus S$. 
\end{definition}

Since $C_1,\dots, C_k$ are all distinct, and the only small components of $S$ are pendant $4$-cycles, either $P=\{e_0\}\cup P_1\cup\{e_1\}$, with $C_1$ being a pendant $4$-cycle, or $C_1,\dots, C_k$ are all large components. In the latter case, we say that $P$ is a \emph{clean} extending path. Notice that a clean extending path might consist of a single edge.

Consider what happens to the cost of $S$ when we add the edges of a clean extending path $P_{uv}=\{e_0\}\cup P_1\cup\{e_1\}\cup P_2\cup \{ e_2\}\dots\cup P_k\cup\{e_k\}$ between nodes $u$ and $v$ of some complex component $C$ of $S$. By doing this, the edges of $P_{uv}$ and of every path in $C$ from $u$ to $v$ become part of a block $B$. Therefore this operation cannot create any new bridge. The cost increases by $k+1$ due to the extra edges of $P_{uv}$. The blocks of every $C_i$ retain their credit, and we collect $1$ credit from every $C_i$. Hence we earn at least $k$ credits from the components along $P_{uv}$. Altogether, the cost increases by at most $1$ based on the above arguments. We call such cost increase, the increase of the cost due to $P_{uv}$, and denote it by $\widetilde \cost(P_{uv})\leq 1$. We remark that in the following we sometimes add $\ell\geq 2$ clean extending paths simultaneously. In that case we will guarantee that they do not share any component $C_i$ (we say that they are \emph{component disjoint}): this way we can guarantee that the overall increase of the cost due to such paths is the sum of their individual cost increases, hence at most $\ell$ (otherwise we might overcount the credits gained from some component). In the following we might add several component disjoint extending paths and remove some edges from $C$. When doing this we will guarantee that the extending paths are still part of some block, despite removing edges of $C$: this way we maintain the invariant that we do not create new bridges and thus the above cost analysis holds.

We next show that we can extend the 3-matching Lemma~\ref{lem:matchingOfSize3} to work with extending paths. If $C$ is a complex component of $S$, we define the graph $G_C$ such that $V(G_C)=V(C)$, and for $u, v\in V(G_C)$, $uv\in E(G_C)$ iff $uv\in E(C)$ or there is an extending path from $u$ to $v$.  

\begin{lemma}\label{lem:3matchingPath}
    Let $C$ be a complex component of $S$, and $(V_1, V_2)$ be a partition of the nodes of $C$ such that $4\leq |V_1|\leq |V_2|$. Then there is a matching of size $3$ in $G_C$ between $V_1$ and $V_2$.
\end{lemma}
\begin{proof}
    We claim $G_C$ has no non-isolating cut. Assume to get a contradiction that there is a non-isolating cut $\{u, v\}$ in $G_C$ splitting $G_C$ into components $C_1, C_2, \dots, C_k, k\geq 2$. Since $G$ is structured, $\{u, v\}$ is not a non-isolating cut of $G$. Therefore, there are at least $2$ components among $C_1, C_2, \dots, C_k$ such that there is a path in $G\setminus\{u, v\}$ between them. Moreover, we can assume such path to be fully outside $C$. Indeed, all edges of $C$ are in $G_C$, so a path $P$ between two components among $C_1, C_2, \dots, C_k$ must contain an edge not in $C$ (otherwise it would also be a path in $G_C$, a contradiction). Thus, the subpath of $P$ that is fully outside $C$ is a path between two components among $C_1, C_2, \dots, C_k$. W.l.o.g. assume there is a path $P$ in $G\setminus\{u, v\}$ between nodes $v_1\in V(C_1)$ and $v_2\in V(C_2)$ not using any edge of $C$. But then there exists an extending path between $v_1$ and $v_2$: this implies that there exists an edge in $G_C\setminus\{u, v\}$ from $v_1$ to $v_2$, a contradiction. The claim follows.
    
    Now assume $|V_1|,|V_2|\ge 4$. Then if there is no matching of size $3$ in $G_C$ between $V_1$ and $V_2$, using König-Egeváry Theorem, there exists a vertex cover $V'$ of the bipartite graph induced by $(V_1,V_2)$ in $G_C$ of size at most $2$. Notice that there is no edge between $V_1\setminus V'$ and $V_2\setminus V'$, where the latter sets have size at least $2$. If $V'$ is of size $2$, then it is a non-isolating cut of $G_C$, a contradiction. If instead $V'=\{v'\}$, then $v'$ is a cut vertex in $G_C$ (hence in $G$), a contradiction.
\end{proof}

For the sake of brevity, when given a complex component $C$ of $S$ and a subgraph $B$ of $C$, we will usually say that we apply Lemma~\ref{lem:3matchingPath} to $B$ instead of to $(V(B), V(C)\setminus V(B))$. Also, because $S$ is canonical, every leaf-block of $C$ has at least $5$ nodes, so we can apply Lemma~\ref{lem:3matchingPath} to any partition $V_1, V_2$ of $V(C)$ if $V_1$ and $V_2$ both contain the nodes of some leaf-block. 

The most complicated part of proving Lemma~\ref{lem:complexComponents} is removing complex components without creating new small components. We do this in Lemma~\ref{lem:bridgeCovering} whose proof is postponed. 
\begin{lemma}\label{lem:bridgeCovering}
    Let $S$ be a canonical 2-edge-cover of a structured graph $G$ whose small components are pendant $4$-cycles and that has at least one complex component. In polynomial time one can find a canonical $2$-edge-cover $S'$ of $G$, whose small components are pendant $4$-cycles, with no more connected components than $S$, with fewer complex components than $S$, and such that $\cost(S')\leq \cost(S)$.
\end{lemma}
After this, the only remaining components are large 2VC components and pendant $4$-cycles. It is then quite simple to remove pendant $4$-cycles and thus obtain a solution satisfying the conditions of Lemma~\ref{lem:complexComponents}. We show this first in the following simple lemma. Then, Lemma~\ref{lem:complexComponents} follows by chaining Lemmas~\ref{lem:bridgeCovering} and \ref{lem:last4Cycles}.

\begin{lemma}\label{lem:last4Cycles}
    Let $S$ be a canonical 2-edge-cover such that the only connected components of $S$ are large 2VC components and pendant $4$-cycles. Then in polynomial time one can find a 2-edge-cover $S'$ such that the only components of $S'$ are large 2VC components and $\cost(S')\leq \cost(S)$.
\end{lemma}
\begin{proof}
    Let $C_1$ be a $4$-cycle of $S$ adjacent only to one connected component $C_2$ of $S$ that is large and 2VC. Let $B$ be the only block of $C_2$. By applying the 3-matching Lemma~\ref{lem:matchingOfSize3} to $C_1$ we can find a matching $\{a_1a_2, b_1b_2\}$, with $a_1b_1\in E(C_1)$ and $a_2,b_2\in V(C_2)$. Set $S':=S\cup\{a_1a_2, b_1b_2\}\setminus\{a_1b_1\}$. In $S'$ the nodes of $C_1$ and $C_2$ belong to a single large component $C_3$ with a single block $B'$. One has 
    $$
    \cost(S)-\cost(S')=\credit(C_1)+\credit(C_2)+\credit(B)-\credit(C_3)-\credit(B')-1=\frac43+1+1-1-1-1>0.   
    $$
    By iterating the above construction we obtain the claim.
\end{proof}

\paragraph{Notation in the figures of Lemma~\ref{lem:bridgeCovering}.} In the coming figures, gray regions represent blocks or $4$-cycles, blue edges are added to $S$ and red edges are being removed from $S$. Blue curved line represent clean extending paths that are added to $S$. Edges in a complex component that become part of a new block are shown in green. The dashed edges represent other edges of $G$ not in $S$.

\begin{proof}[Proof of Lemma \ref{lem:bridgeCovering}]
    Let $C$ be any complex component of $S$. We will show how to modify $S$ into a canonical 2-edge-cover $S'$ without creating new connected components and in a way that $V(C)$ is contained in a connected component $C'$ of $S'$, and either $S'$ has fewer connected components than $S$ or $C'$ has fewer bridges or fewer blocks than $C$. We then repeat this operation on $C'$ until it becomes a (large) 2VC component.

    For every $a, b\in C$ let $P^C_{ab}$ be any path in $C$ between them. We consider the following cases, assuming at each case that the previous ones do not hold.

\begin{caseanalysis}
\case{For some leaf-block $B$, there is a non-clean extending path $P_{xy}$ from $x\in V(B)\setminus\{u\}$ to $y\in V(C)\setminus V(B)$, where $u$ is the only cut vertex of $B$.}\label{case:nonCleanPath}\hfill\\
This case is illustrated in Figure~\ref{fig:caseNoCleanB}. Let $C'$ be the $4$-cycle corresponding to $P_{xy}$ and let $xa, a\in V(C')$, be the first edge of $P_{xy}$. Apply the 3-matching Lemma~\ref{lem:matchingOfSize3} to $C'$ to get a matching $\{v_1w_1, v_2w_2, v_3w_3\}$, with $w_i\in C, v_i\in C'$ for all $i$. W.l.o.g. we can assume  $w_1=x$: indeed if $w_1\neq x$, we can add $ax$ to the matching and remove $aw_i$ if any. Notice that at least one of $v_2$ or $v_3$ is adjacent to $v_1$, say $v_1v_2\in E(C')$ w.l.o.g. Set $S':=S\cup\{v_1w_1, v_2w_2\}\setminus\{v_1v_2\}$. In $S'$, the path $E(C')\cup\{v_1w_1, v_2w_2\}\setminus\{v_1v_2\}$ is merged with $B$ (and potentially more edges of $C$) into a single block $B'$. Hence one has
\begin{align*}
\cost(S)-\cost(S')\geq \credit(B)+\credit(C')-\credit(B')-1=1+4/3-1-1>0.
\end{align*}

\begin{figure}
    \centering
    \includegraphics[scale=0.8]{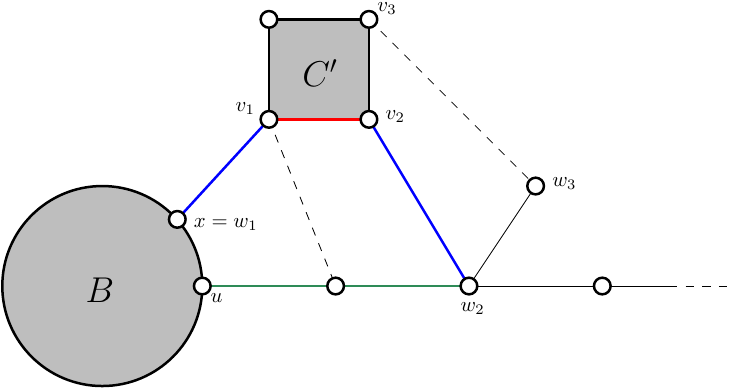}
    \caption{Lemma~\ref{lem:bridgeCovering} Case~\ref{case:nonCleanPath}.}
    \label{fig:caseNoCleanB}
\end{figure}

    We will use the next remark repeatedly in the rest of the cases.
    \begin{remark}\label{rem:lemma18}
    Let $B$ be any leaf-block with cut vertex $u$, and assume that Case~\ref{case:nonCleanPath} does not hold.     
    Apply Lemma~\ref{lem:3matchingPath} to $B$. There are two extending paths with distinct endpoints from $v_1, v_2\in V(B)\setminus\{u\}$ to $w_1, w_2\in V(C)\setminus V(B)$. Since Case~\ref{case:nonCleanPath} does not hold, these paths are clean. Given $x\in V(C)\setminus V(B)$ at least one of those paths goes from $a$ to $b$, where $a\in V(B)\setminus\{u\}$, $b\in V(C)\setminus (V(B)\cup\{x\})$.
    \end{remark}
   
\case{For some block $B$, there is a clean extending path $P_{xy}$ from $x\in V(B)$ to $y\in V(C)\setminus V(B)$ such that there is a path $P^C_{xy}$ containing an edge of $B$ and either an edge of another block $B'$ or at least $4$ bridges.}\label{case:cleanToBlock}\hfill\\
This case is illustrated in Figure~\ref{fig:case2BlocksPath}. Set $S':=S\cup E(P_{xy})$. $P_{xy}, B$ and all edges of $P^C_{xy}$ are merged into a single block $B''$ of $S'$. If $P^C_{xy}$ contains the edge of another block $B'$ then $B'$ is also part of $B''$ in $S'$, and it brings $1$ credit. If $P^C_{xy}$ contains at least $4$ bridges each of them brings $1/4$ credits, so we collect $1$ credit from the edges of $P^C_{xy}$ as well. Altogether 
    \begin{align*}
        \cost(S)-\cost(S')\geq \credit(B) + \credit(P^C_{xy}) - \credit(B'')  - \widetilde\cost(P_{xy})\geq 1 +1 -1-1 =0.
    \end{align*}

\begin{figure}
    \centering
    \includegraphics[scale=0.8]{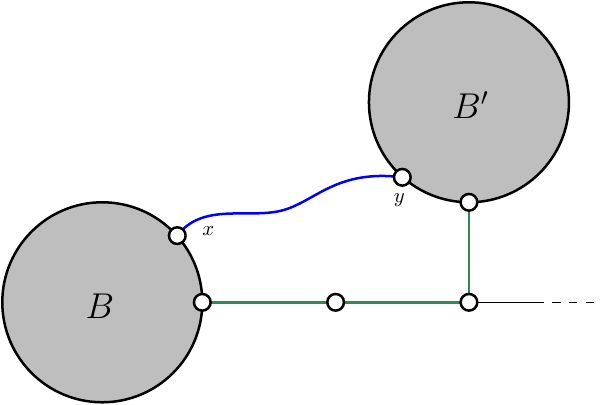}
    \caption{Lemma~\ref{lem:bridgeCovering} Case~\ref{case:cleanToBlock}.}
    \label{fig:case2BlocksPath}
\end{figure}

\begin{remark}\label{rem:atLeastOneBridge}
Assuming that Cases~\ref{case:nonCleanPath} and~\ref{case:cleanToBlock} do not hold, every leaf-block $B$ must have at least one bridge incident to them. Otherwise, let $u$ be the cut vertex of $B$. By Remark \ref{rem:lemma18}, there is a clean extending path from a vertex $v\in V(B)\setminus\{u\}$ to some $w\in V(C)\setminus V(B)$. Since no bridge is incident to $u$, $P^C_{wu}$ includes the edge of another block $B'$, and thus $P^C_{vw}$ includes edges of both $B$ and $B'$, which is excluded by Case~\ref{case:cleanToBlock}. 
\end{remark}

Let $C^*$ be the block-cutpoint graph (see Section~\ref{sec:preliminaries}) of $C$, and take a longest path $P^*$ in the tree $C^*$. Note that the endpoints of this paths are two leaf-blocks. Let $B$ be the block corresponding to one of those leaves. By Remark \ref{rem:atLeastOneBridge} and because $P^*$ is a longest path, $B$ must have exactly one bridge $u_0u_1$, $u_0\in V(B)$, incident to it. $P^*$ is the path $b_1c_1b_2c_2b_3c_3\dots$, where $b_1, b_2,\dots$ are blocks or bridges of $C$, $c_1, c_2,\dots$ are cut vertices of $C$ and $b_1, c_1, b_2$ correspond to $B, u_0$ and $u_0u_1$, respectively. Consider the maximal sequence $b_2, b_3,\dots$ such that $b_i$ is a bridge of $C$ for all $i\geq 2$, and let $u_0u_1u_2u_3\dots$ be the path of bridges corresponding to that sequence. For every $i\geq 0$, let $A(u_i)$ be the set of nodes (other than $u_i$) with paths in $C$ to $u_i$ that do not go through $u_{i-1}$ or $u_{i+1}$. Observe that the sets $A(u_i)$ are disjoint and that $V(B)\setminus\{u_0\}\subseteq A(u_0)$. An example of how $C$ and $P^*$ might look like is shown in Figure~\ref{fig:longestPathP}.

\begin{figure}[ht]
    \centering
    \includegraphics[scale=0.8]{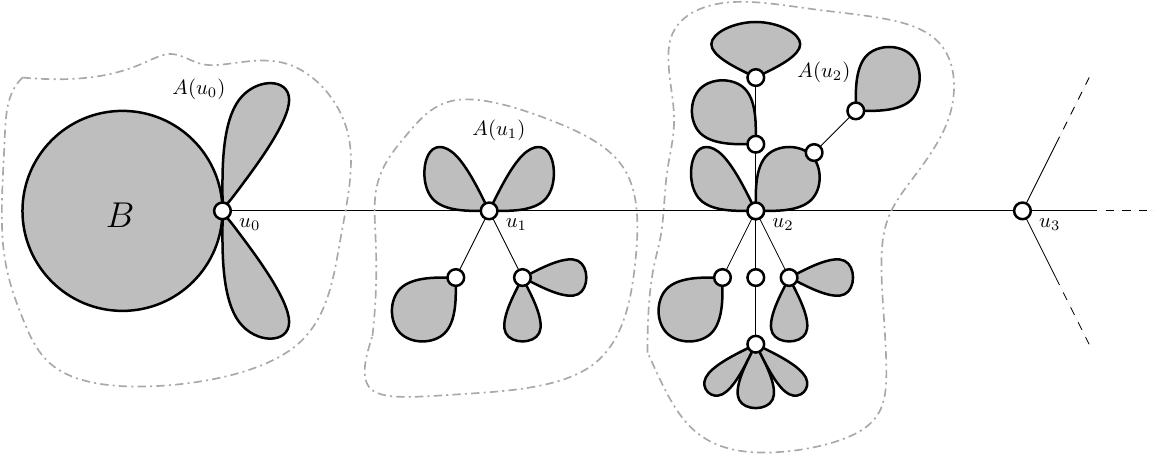}
    \caption{A representation of $C$ with respect to the path $P^*$. Recall that, because of Remark~\ref{rem:atLeastOneBridge}, every leaf-block must have a bridge incident to it.}
    \label{fig:longestPathP}
\end{figure}

\case{There is clean extending path $P_{vw}$ from $v\in V(B)\setminus\{u_0\}$ to $w\in A(u_1)$.}\label{case:cleanToAu1}\hfill\\
This case is illustrated in Figure~\ref{fig:blockExtendingCase11}. No path $P^C_{wu_0}$ can include the edge of a block distinct from $B$, because of Case~\ref{case:cleanToBlock}. Thus $P^C_{wu_0}$ consists only of bridges. Also, since $P^*$ is a longest path, $P^C_{wu_1}$ consists of a single bridge $u_1w$ and $w\in V(B')$, where $B'$ is a leaf-block of $C$ distinct from $B$. Apply Remark \ref{rem:lemma18} to $B'$ to find a clean extending path $P_{xy}$ from $x\in V(B')\setminus\{w\}$, to $y\in V(C)\setminus V(B')$. Let us first argue that the paths $P_{vw}$ and $P_{xy}$ are component disjoint. If that is not the case then there is a clean extending path $P_{vx}$ from $v$ to $x$, but since every $P^C_{vx}$ includes at least an edge of $B$ and an edge of $B'$ this is excluded by Case~\ref{case:cleanToBlock}.
   
Notice that every path $P^C_{yu_0}$ must not use $u_1w$. Otherwise, since $P^*$ is a longest path, it must be that $y\in B''$, where $B''$ is another leaf-block with $w\in V(B'')$, and $P^C_{xy}$ contains the edges of two distinct blocks, a contradiction by Case~\ref{case:cleanToBlock}. Set $S':=S\cup E(P_{vw})\cup E(P_{xy})\setminus\{u_1w\}$, so that $B$ and $B'$ are merged into a single block $B''$ by the paths $P_{vw}$ and $P_{xy}\cup P^C_{yu_0}$. One has 
\begin{align*}
    \cost(S)-\cost(S') &\geq \credit(B)+\credit(B')+1-\credit(B'')-\widetilde\cost(P_{vw})-\widetilde\cost(P_{xy})\\
    & \geq 1+1+1-1-1-1=0.
\end{align*}

\begin{figure}
    \centering
    \includegraphics[scale=0.8]{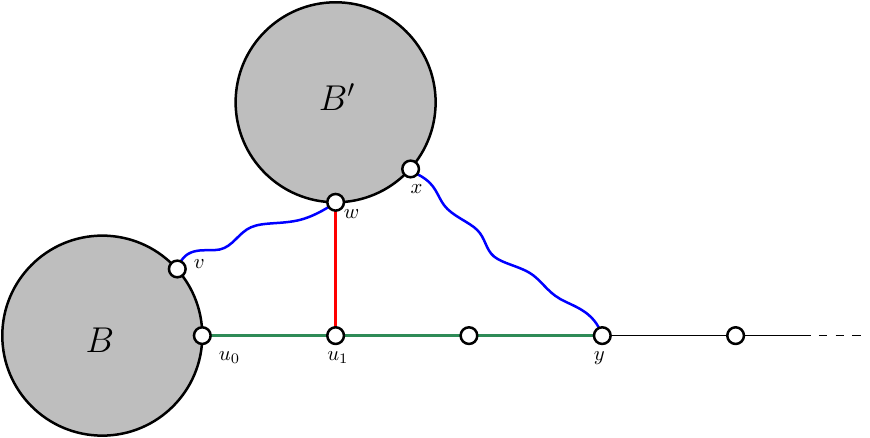}
    \caption{Lemma~\ref{lem:bridgeCovering} Case~\ref{case:cleanToAu1}.}
    \label{fig:blockExtendingCase11}
\end{figure}

\case{There is a clean extending path $P_{vw}$ from $v\in V(B)\setminus\{u_0\}$ to $w\in A(u_2)$.}\label{case:cleanToAu2}\hfill\\
By Case~\ref{case:cleanToBlock} it must be that every $P^C_{wu_0}$ contains no edge of a block and $|E(P^C_{wu_0})|\leq 3$. Thus there is a bridge $u_2w$ in $C$ and $P^C_{wu_0}=wu_2u_1u_0$.

\subcase{There is another bridge $ww', w'\neq u_2$.}\label{case:cleanToAu2bridge}\hfill\\
This case is illustrated in Figure~\ref{fig:blockExtendingCase21}. Since $P^*$ is a longest path, $w'\in V(B')$, where $B'$ is a leaf-block of $C$ distinct from $B$. Apply Remark~\ref{rem:lemma18} to $B'$ to find a clean extending path $P_{xy}$ from $x\in V(B')\setminus\{w'\}$ to $y\in V(C)\setminus(V(B')\cup \{w\})$. Notice that $P_{xy}$ and $P_{vw}$ are component disjoint, otherwise there is a clean extending path $P_{vx}$ from $v$ to $x$, which is excluded by Case~\ref{case:cleanToBlock}. 
    
Let $C_1$ and $C_2$ be the connected components resulting from removing the bridge $u_2w$, with $V(B)\subset V(C_1)$. If $y\in V(C_2)$ then either $y\in V(B'')$, where $B''$ is another leaf-block with $w'\in V(B'')$, and $P^C_{xy}$ contains the edges of two distinct blocks (a contradiction by Case~\ref{case:cleanToBlock}), or we are in a case symmetric to Case~\ref{case:cleanToAu1} (with $B'$ in place of $B$). So we can assume $y\in V(C_1)$ and therefore $w\notin V(P^C_{yu_0})$. Set $S':=S\cup E(P_{vw})\cup E(P_{xy})$, so that in $S'$, $B$ and $B'$ are merged into a block $B''$ of $S'$ through the paths $P_{vw}\cup \{ww'\}$ and $P_{xy}\cup P^C_{yu_0}$. Observe that the path $P^C_{wu_0}$ is also part of $B''$ in $S'$. One has 
\begin{align*}
    \cost(S)-\cost(S') &\geq \credit(B)+\credit(B')+\credit(ww')+\credit(wu_2)+\credit(u_2u_1)+\credit(u_1u_0)\\
    &-\credit(B'')-\widetilde\cost(P_{vw})-\widetilde\cost(P_{xy}) \geq 1+1+4\cdot 1/4-1-1-1=0.
\end{align*}

\begin{figure}
    \centering
    \includegraphics[scale=0.8]{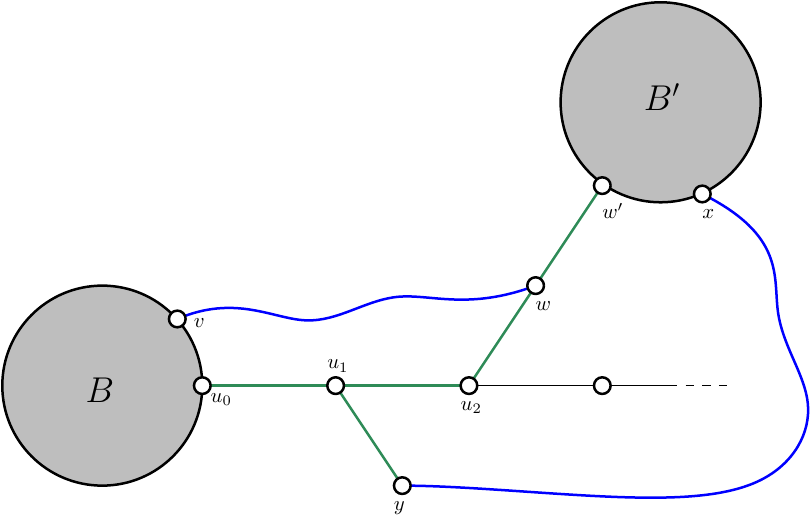}
    \caption{Lemma~\ref{lem:bridgeCovering} Case~\ref{case:cleanToAu2bridge}.}
    \label{fig:blockExtendingCase21}
\end{figure}

\subcase{There is no bridge $ww', w'\neq u_2$.}\label{case:cleanToAu2nobridge}\hfill\\
This case is illustrated in Figure~\ref{fig:blockExtendingCase22}.  Using the fact that $P^*$ is a longest path and Remark \ref{rem:atLeastOneBridge}, it must be that $w\in V(B')$ for some leaf-block $B'$ of $C$ distinct from $B$. Apply Remark \ref{rem:lemma18} to $B'$ to find a clean extending path $P_{xy}$ from $x\in V(B')\setminus \{w\}$ to $y\in V(C)\setminus V(B')$. As in the previous case, $P_{xy}$ and $P_{vw}$ are component disjoint, otherwise there is a clean extending path $P_{vx}$ from $v$ to $x$, which is excluded by Case~\ref{case:cleanToBlock}.
    
If $wu_2\in E(P^C_{yu_0})$ for some path $P^C_{yu_0}$, it must be that $y\in V(B'')$, where $B''$ is another leaf-block of $C$ with $w\in V(B'')$, a contradiction by Case~\ref{case:cleanToBlock}. Set $S':=S\cup E(P_{vw})\cup E(P_{xy})\setminus\{u_2w\}$, so that $B$ and $B'$ are merged into a single block $B''$ of $S'$ through the paths $P_{vw}$ and $P_{xy}\cup P^C_{yu_0}$. Then we get  
\begin{align*}
    \cost(S)-\cost(S') &\geq \credit(B)+\credit(B')+1-\credit(B'')-\widetilde\cost(P_{vw})-\widetilde\cost(P_{xy})\\
    & \geq 1+1+1-1-1-1=0.
\end{align*}

\begin{figure}
    \centering
    \includegraphics[scale=0.8]{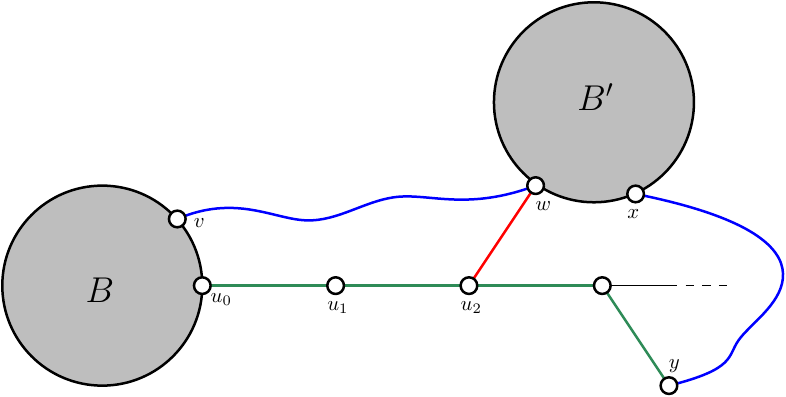}
    \caption{Lemma~\ref{lem:bridgeCovering} Case~\ref{case:cleanToAu2nobridge}.}
    \label{fig:blockExtendingCase22}
\end{figure}

\case{There are two clean extending paths $P_{v_1u_2}$ and $P_{v_2u_3}$ from $v_1, v_2\in V(B)\setminus\{u_0\}$, $v_1\neq v_2$ to $u_2$ and $u_3$.}\label{case:twoPaths}\hfill\\
Notice that $P_{v_1u_2}$ and $P_{v_2u_3}$ might not be component disjoint. Removing the bridge $u_2u_3$ splits $C$ into $2$ connected components $C_1, C_2$, with $V(B)\subseteq V(C_1)$.

\subcase{There is a non-clean extending path $P_{ab}$ from $a\in V(C_1)$ to $b\in V(C_2)$.}\label{case:twoPathsNonClean}\hfill\\
This case is illustrated in Figure~\ref{fig:blockExtendingCase3pendant}. Let $C'$ be the $4$-cycle corresponding to $P_{ab}$ and let $aa'$ and $bb'$ be the first and last edges of $P_{ab}$. Apply the 3-matching Lemma~\ref{lem:matchingOfSize3} to $C'$ to get a matching $\{x_1y_1, x_2y_2, x_3y_3\}$, with $x_i\in C, y_i\in C'$ for all $i$. 

Let us prove that we can assume w.l.o.g. that $x_1\in V(C_1)$, $x_2\in V(C_2)$ and $y_1y_2\in E(C')$. First notice that we can assume w.l.o.g. that neither $x_i\in V(C_1)$ for all $i\in\{1, 2, 3\}$ nor $x_i\in V(C_2)$ for all $i\in\{1, 2, 3\}$. Indeed, if  $x_i\in V(C_1)$ for $i\in\{1, 2, 3\}$ add $bb'$ to the matching and remove $x_ib'$ if any. Symmetrically if $x_i\in V(C_2)$ for $i\in\{1, 2, 3\}$ add $aa'$ to the matching and remove $x_ia'$ if any. Now assume $x_1, x_3\in V(C_1)$ and $x_2\in V(C_2)$. Since $C'$ is a $4$-cycle, at least one of $y_1, y_3$, say w.l.o.g. $y_1$, is adjacent to $y_2$ in $C'$, so $x_1\in V(C_1), x_2\in V(C_2)$ and $y_1y_2\in E(C')$. The case where two of $x_1, x_2, x_3$ belong to $V(C_2)$ is symmetric. From now on we assume that $x_1\in V(C_1)$, $x_2\in V(C_2)$ and $y_1y_2\in E(C')$. 
    
Set $S':=S\cup E(P_{v_2u_3})\cup\{x_1y_1, x_2y_2\}\setminus\{y_1y_2\}$. Notice that since $x_1\in V(C_1), x_2\in V(C_2)$, there are vertex disjoint paths $P^C_{x_1u_2}$ and $P^C_{x_2u_3}$. Adding $P_{v_2u_3}$ merges $B$ and $P^C_{u_0u_3}=u_0u_1u_2u_3$ into a single block, and adding $E(C')\cup\{x_1y_1, x_2y_2\}\setminus\{y_1y_2\}$ merges that block and the paths $P^C_{x_1u_2}$, $P^C_{x_2u_3}$ into a new block $B'$ of $S'$. One has
\begin{align*}
    \cost(S)-\cost(S') &\geq \credit(B)+\credit(C') +\credit(u_0u_1)+\credit(u_1u_2)+\credit(u_2u_3)\\
    &- \credit(B') - \widetilde\cost(P_{v_2u_3}) -1 \geq 1 + 4/3 + 3\cdot 1/4 - 1 - 1 - 1 >0.
\end{align*}

\begin{figure}
    \centering
    \includegraphics[scale=0.8]{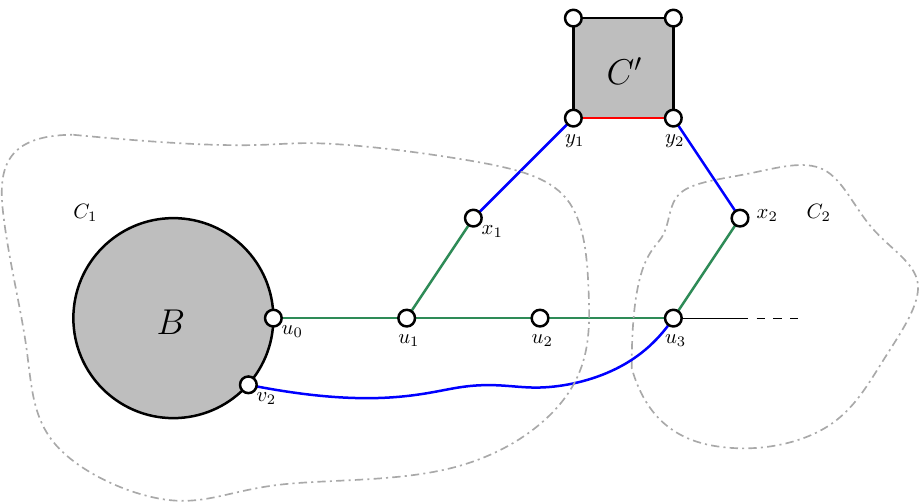}
    \caption{Lemma~\ref{lem:bridgeCovering} Case~\ref{case:twoPathsNonClean}.}
    \label{fig:blockExtendingCase3pendant}
\end{figure}

\subcase{All extending paths between nodes of $C_1$ and $C_2$ are clean.}\label{case:twoPathsClean}\hfill\\
This case is illustrated in Figure~\ref{fig:caseu3}. Apply Lemma~\ref{lem:3matchingPath} to $(V(C_1), V(C_2))$. By the hypothesis of this case, there must be at least one clean extending path $P_{xy}$ from $x\in V(C_2)\setminus\{u_3\}$ to $y\in V(C_1)\setminus\{u_0\}$. Since $x\neq u_3$, $P^C_{xu_3}$ is not empty, and it contains at least one edge $e$. Thus, $|E(P^C_{xu_0})|\geq 4$, so $y\notin A(u_0)$ by Case~\ref{case:cleanToBlock}, and thus $y\in A(u_1)\cup A(u_2)\cup\{u_1, u_2\}$. Observe that the paths $P_{xy}$ and $P_{v_1u_2}$ are component disjoint. Indeed, otherwise there is a clean extending path $P_{v_1x}$ from $v_1$ to $x$, and there is a path $P^C_{v_1x}$ that includes an edge of $B$ and at least $4$ other edges not in $B$ (since $|E(P^C_{xu_0})|\geq 4$); this is excluded by Case~\ref{case:cleanToBlock}. By a similar argument $P_{xy}$ and $P_{v_2u_3}$ are also component disjoint. 
    
If $y\in A(u_1)\cup\{u_1\}$ set $S':=S\cup E(P_{v_1u_2})\cup E(P_{xy})\setminus \{u_1u_2\}$. Notice that $u_1u_2$ is neither in $P^C_{u_1y}$ nor in $P^C_{xu_3}$. The path $\{u_0u_1\}\cup P^C_{u_1y}\cup P_{xy}\cup P^C_{xu_3}\cup\{u_2u_3\}\cup P_{v_1u_2}$ is merged with $B$ into a single block $B'$ of $S'$. If $e$ is the edge of a block then that block is also merged into $B'$, bringing $1$ credit, and if it is a bridge it brings $1/4$ credits, so we get at least $1/4$ credits from the edges of $P^C_{xu_3}$. One has
\begin{align*}
    \cost(S)-\cost(S') &\geq \credit(B)+\credit(u_0u_1)+\credit(u_1u_2)+\credit(u_2u_3)+1/4+1\\
    &-\credit(B')-\widetilde\cost(P_{v_1u_2})-\widetilde\cost(P_{xy})\geq 1+4\cdot 1/4+1-1-1-1=0.
\end{align*}

If $y\in A(u_2)\cup\{u_2\}$ set $S':=S\cup E(P_{v_2u_3})\cup E(P_{xy})\setminus \{u_2u_3\}$. Notice that $u_2u_3$ is neither in $P^C_{u_2y}$ nor in $P^C_{xu_3}$. The path $\{u_0u_1\}\cup \{u_1u_2\}\cup P^C_{u_2y}\cup P_{xy} \cup P^C_{xu_3}\cup P_{v_2u_3}$ is merged with $B$ into a single block $B'$ of $S'$. As before, if $e$ is the edge of a block then that block is also merged into $B'$, bringing $1$ credit, and if it is a bridge it brings $1/4$ credits, so we get at least $1/4$ credits from the edges of $P^C_{xu_3}$. One has 
    \begin{align*}
        \cost(S)-\cost(S') &\geq \credit(B)+\credit(u_0u_1)+\credit(u_1u_2)+\credit(u_2u_3)+1/4+1\\
        &-\credit(B')-\widetilde\cost(P_{v_2u_3})-\widetilde\cost(P_{xy})\geq 1+4\cdot 1/4+1-1-1-1=0.
    \end{align*}

\begin{figure}
    \centering
    \includegraphics[scale=0.8]{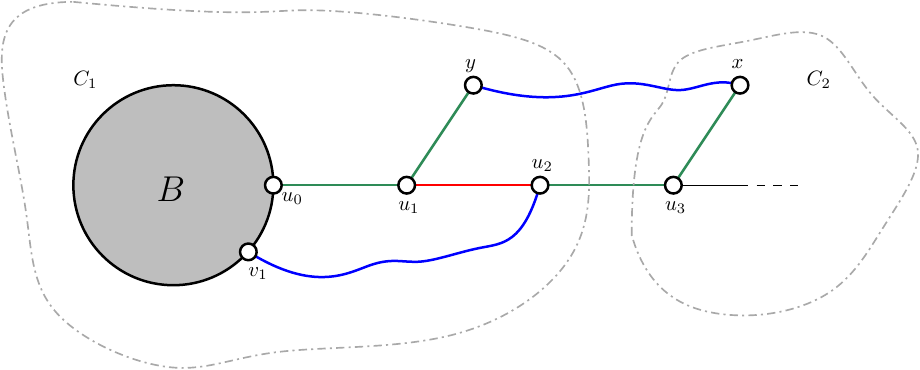}
    \includegraphics[scale=0.8]{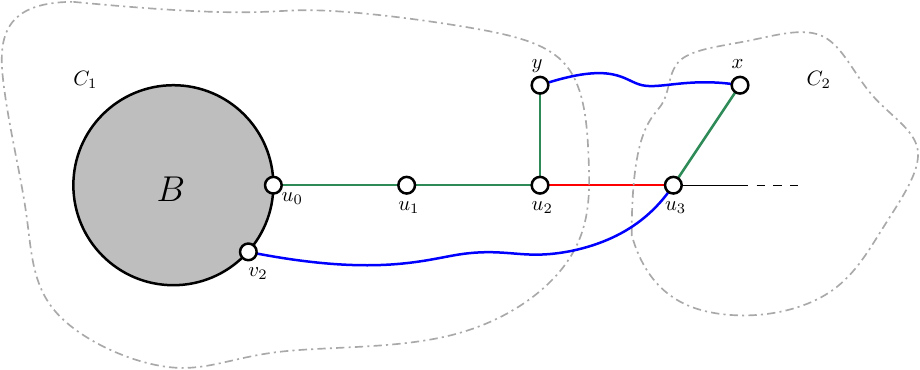}
    \caption{Lemma~\ref{lem:bridgeCovering} Case~\ref{case:twoPathsClean}. We show the two possibilities, $y\in A(u_1)\cup\{u_1\}$ and $y\in A(u_2)\cup\{u_2\}$.}
    \label{fig:caseu3}
\end{figure}

\begin{remark}\label{rem:extendingPathEndpoints}
    Assume that all previous cases do not hold. Then, if there is a clean extending path $P_{vw}$ from $v\in V(B)\setminus\{u_0\}$ to $w\in V(C)\setminus V(B)$ one has $w\in \{u_1, u_2, u_3\}$. To see this, notice that by Case~\ref{case:cleanToBlock} and since the path from every $u\in A(u_0)$ to $u_0$ contains the edge of a block, $w\notin A(u_0)$. By Cases~\ref{case:cleanToAu1} and~\ref{case:cleanToAu2} $w\notin A(u_1)\cup A(u_2)$, and by Case~\ref{case:cleanToBlock} $w\in A(u_0)\cup A(u_1)\cup A(u_2)\cup \{u_1, u_2, u_3\}$, because otherwise $|E(P^C_{wu_0})|\geq 4$.
\end{remark}

\case{None of the previous cases hold.}\hfill\\
Many subcases of this case are analogous to previous cases; however, we include them all here for the sake of completeness. Apply Lemma~\ref{lem:3matchingPath} to $B$ to find a matching $\{e_1, e_2, e_3\}$ in $G_C$. Since $B$ is a leaf-block, we can assume w.l.o.g. that $e_1, e_2$ correspond to extending paths with distinct endpoints from $V(B)\setminus\{u_0\}$ to $V(C)\setminus V(B)$. Since Case~\ref{case:nonCleanPath} does not hold, those paths are clean. By Remark~\ref{rem:extendingPathEndpoints} and the fact that Case~\ref{case:twoPaths} does not hold, at least one of those paths has an endpoint in $u_1$ and the other in either $u_2$ or $u_3$. This implies that $e_1, e_2$ correspond to clean extending paths $P_{v_1u_1}, P_{v_2w_2}$ with distinct endpoints from $v_1, v_2\in V(B)\setminus\{u_0\}$ to $u_1$ and $w_2\in\{u_2, u_3\}$, resp.

\subcase{There is a leaf-block $B'$ with $V(B')\subset A(u_0)$ distinct from $B$.}\label{case:degenerateA0}\hfill\\
This case is illustrated in Figure~\ref{fig:caseA06}. Apply Remark \ref{rem:lemma18} to $B'$ to find a clean extending path $P_{xy}$ from $x\in V(B')\setminus\{u_0\}$, to $y\in V(C)\setminus V(B')$. Note that by Case~\ref{case:cleanToBlock} we have $y\notin A(u_0)$ since otherwise $P^C_{yu_0}$ would include the edge of a block distinct from $B'$. Set $S':=S\cup E(P_{v_1u_1})\cup E(P_{xy})\setminus\{u_0u_1\}$, so that $B$ and $B'$ are merged into a single block $B''$ of $S'$ through the paths $P_{v_1u_1}\cup P^C_{yu_1}$ and $P_{xy}$. Notice that by Case~\ref{case:cleanToBlock}, $P_{v_1u_1}$ and $P_{xy}$ are component disjoint. One has
    \begin{align*}
        \cost(S)-\cost(S') &\geq \credit(B)+\credit(B')+1-\credit(B'')- \widetilde\cost(P_{v_1u_1})-\widetilde\cost(P_{xy})\\
        & \geq 1 + 1 + 1 - 1 - 1 - 1 = 0.
    \end{align*}

\begin{figure}
    \centering
    \includegraphics[scale=0.8]{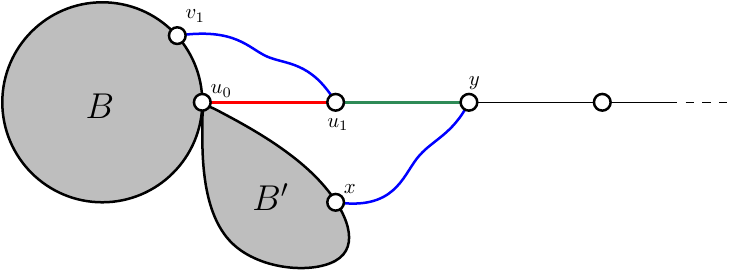}
    \caption{Lemma~\ref{lem:bridgeCovering} Case~\ref{case:degenerateA0}.}
    \label{fig:caseA06}
\end{figure}

\subcase{There is a non-clean extending path $P_{u_0w}$ from $u_0$ to $w\in V(C)\setminus V(B)$.}\label{case:degenerateNonClean}\hfill\\ 
This case is illustrated in Figure~\ref{fig:caseNoCleanB6}. Let $C'$ be the pendant $4$-cycle corresponding to $P_{u_0w}$. Apply the 3-matching Lemma~\ref{lem:matchingOfSize3} to $C'$ to get a matching $\{x_1y_1, x_2y_2, x_3y_3\}$ with $x_i\in C, y_i\in C'$ for all $i$. W.l.o.g. we can assume $x_1=u_0$ and $y_1y_2\in E(C')$. If $x_1\neq u_0$ we can add the first edge of $P_{u_0w}$ to the matching and remove any edge sharing an endpoint with it. Since one of $y_2$ or $y_3$ must be adjacent to $y_1$ we can also assume $y_1y_2\in E(C')$. By exclusion of Case~\ref{case:degenerateA0}, every path in $C$ from $u_1$ to $x_2$ does not include the edge $u_0u_1$. Set $S':=S\cup E(P_{v_1u_1})\cup\{u_0y_1, x_2y_2\}\setminus\{u_0u_1, y_1y_2\}$. Observe that $P_{v_1u_1}$, $E(C')\cup\{u_0y_1, x_2y_2\}\setminus\{y_1y_2\}$, and $P^C_{x_2u_1}$ form a path from $v_1$ to $u_0$ such that its inner vertices are not in $B$, and thus in $S'$ they are merged with $B$ into a single block $B'$ of $S'$. One has 
\begin{align*}
    \cost(S)-\cost(S') &\geq \credit(B)+\credit(C')-\credit(B')-\widetilde\cost(P_{v_1u_1}) \geq 1+4/3-1 - 1 > 0.
\end{align*}

\begin{cfigure}
    \centering
    \includegraphics[scale=0.8]{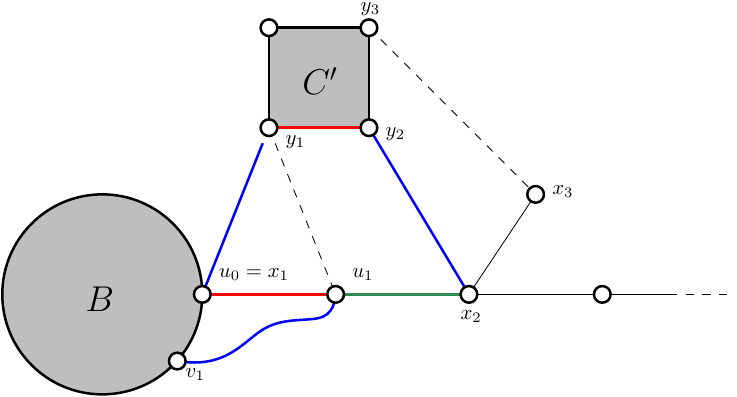}
    \caption{Lemma~\ref{lem:bridgeCovering} Case~\ref{case:degenerateNonClean}.}
    \label{fig:caseNoCleanB6}
\end{cfigure}

\begin{remark}\label{rem:degenerateCompDisjoint}
    For every clean extending path $P_{u_0w}$ from $u_0$ to $w\in V(C)\setminus (V(B)\cup\{u_1, w_2\})$, $P_{u_0w}$ and $P_{v_1u_1}$ are component disjoint. Indeed, otherwise there exists a clean extending path $P_{v_1w}$ from  and $v_1$ to $w$, and by Remark~\ref{rem:extendingPathEndpoints}, $w\in \{u_2, u_3\}$. This, together with the path $P_{v_2w_2}$, constitute a contradiction to the fact that Case~\ref{case:twoPaths} does not hold. 
\end{remark}

\subcase{There is a clean extending path $P_{u_0w}$ from $u_0$ to $w\in V(C)\setminus (V(B)\cup\{u_1, w_2\})$ such that there is a path $P^C_{u_0w}$ containing the edge of another block $B'$ or at least $4$ bridges.}\label{case:degenerateCleanToBlock}\hfill\\
This case is illustrated in Figure~\ref{fig:case2BlocksPath6}. By Remark~\ref{rem:degenerateCompDisjoint}, $P_{v_1u_1}$ and $P_{u_0w}$ are component disjoint. Set $S':=S\cup E(P_{v_1u_1})\cup E(P_{u_0w})\setminus \{u_0u_1\}$. Notice that $P^C_{u_1w}$ does not include the edge $u_0u_1$, since by Case~\ref{case:degenerateA0}, $A(u_0)=V(B)\setminus\{u_0\}$. The paths $P_{u_0w}\cup P^C_{u_1w}$ and $P_{v_1u_1}$ are merged with $B$ into a single block $B''$ of $S'$. All edges of $P^C_{u_0w}$ (except $u_0u_1$) are also included in $B''$. If one of those edges is the edge of another block it gets merged as well in $B''$, and if not we get the credit from $4$ bridges. In both cases we get at least $1$ credit. One has
\begin{align*}
    \cost(S)-\cost(S') &\geq 1 + 1 - \widetilde\cost(P_{v_1u_1})-\widetilde\cost(P_{u_0w}) \geq 1 + 1 - 1-1  = 0.
\end{align*}

\begin{figure}
    \centering
    \includegraphics[scale=0.8]{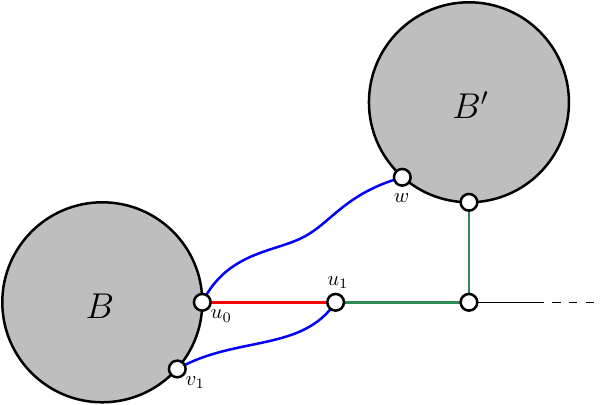}
    \caption{Lemma~\ref{lem:bridgeCovering} Case~\ref{case:degenerateCleanToBlock}.}
    \label{fig:case2BlocksPath6}
\end{figure}

\subcase{There is a clean extending path $P_{u_0w}$ from $u_0$ to $w\in A(u_1)$.}\label{case:degenerateAu1}\hfill\\
This case is illustrated in Figure~\ref{fig:blockExtendingcasewu1}. By Case~\ref{case:degenerateCleanToBlock} and the fact that $u_1, w_2\notin A(u_1)$, no path $P^C_{u_0w}$ contains the edge of another block $B'$. Thus, $P^C_{u_0w}$ consists only of bridges. Also, since $P^*$ is a longest path, $P^C_{wu_1}$ consists of a single bridge $u_1w$ and $w\in V(B')$, where $B'$ is a leaf-block of $C$ distinct from $B$. Apply Remark~\ref{rem:lemma18} to find a clean extending path $P_{xy}$ from $x\in V(B')\setminus\{w\}$ to $y\in V(C)\setminus V(B')$.

Let us first argue that the paths $P_{u_0w}, P_{v_1u_1}$ and $P_{xy}$ are pair-wise component disjoint. By Remark~\ref{rem:degenerateCompDisjoint}, $P_{u_0w}$ and $P_{v_1u_1}$ are component disjoint. If $P_{v_1u_1}$ and $P_{xy}$ are not component disjoint, then there is a clean extending path $P_{v_1x}$ from $v_1\in V(B)\setminus\{u_0\}$ to $x$, a contradiction by Case~\ref{case:cleanToBlock}. If $P_{u_0w}$ and $P_{xy}$ are not component disjoint then there is a clean extending path $P_{u_0x}$ from $u_0$ to $x$. Since every $P^C_{u_0x}$ includes at least an edge of $B'$ this is excluded by Case~\ref{case:degenerateCleanToBlock}.

Notice that every path $P^C_{yu_1}$ must not use $u_1w$. Otherwise, since $P^*$ is a longest path, it must be that $y\in B''$, where $B''$ is another leaf-block with $w\in V(B'')$, and $P^C_{xy}$ contains the edges of two distinct blocks, a contradiction by Case~\ref{case:cleanToBlock}. Also, every path $P^C_{yu_1}$ does not use $u_0u_1$, because by exclusion of Cases~\ref{case:degenerateA0} and~\ref{case:degenerateCleanToBlock}, $y\notin A(u_0)\cup \{u_0\}$. Set $S':=S\cup E(P_{v_1u_1})\cup E(P_{u_0w})\cup E(P_{xy})\setminus\{u_0u_1, u_1w\}$, so that $B$ and $B'$ are merged into a single block $B''$ through the paths $P_{u_0w}$ and $P_{xy}\cup P^C_{yu_1}\cup P_{u_1v_1}$. One has 
\begin{align*}
    \cost(S)-\cost(S') &\geq \credit(B)+\credit(B')+2-\credit(B'')-\widetilde\cost(P_{v_1u_1})-\widetilde\cost(P_{u_0w})-\widetilde\cost(P_{xy})\\
    & \geq 1+1+2-1-1-1-1=0.
\end{align*}

\begin{figure}
    \centering
    \includegraphics[scale=0.8]{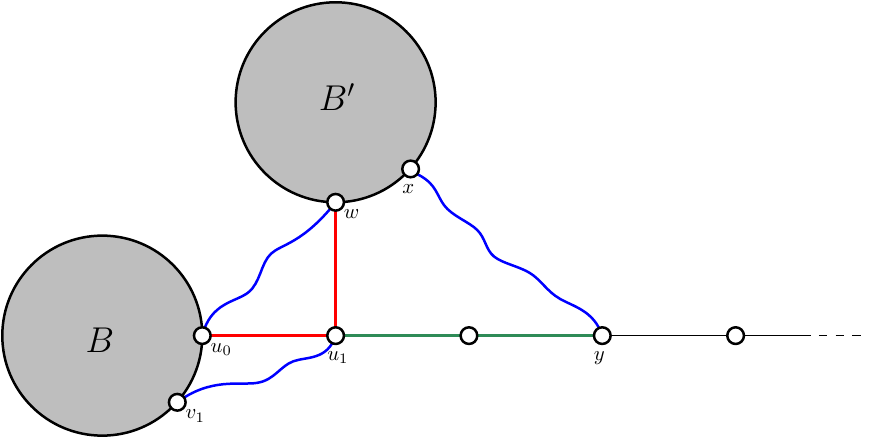}
    \caption{Lemma~\ref{lem:bridgeCovering} Case~\ref{case:degenerateAu1}.}
    \label{fig:blockExtendingcasewu1}
\end{figure}

\subcase{There is clean extending path $P_{u_0w}$ from $u_0$ to $w\in A(u_2)$.}\label{case:degenerateAu2}\hfill\\
By Case~\ref{case:degenerateCleanToBlock} and the fact that $u_1, w_2\notin A(u_2)$ it must be that every $P^C_{wu_0}$ contains no edge of a block and $|E(P^C_{wu_0})|\leq 3$. Thus there
is a bridge $u_2w$ in $C$ and $P^C_{wu_0} = wu_2u_1u_0$.

\subsubcase{There is another bridge $ww', w'\neq u_2$.}\label{case:degenerateAu2bridge}\hfill\\
This case is illustrated in Figure~\ref{fig:blockExtendingCasewu21}. Since $P^*$ is a longest path, $w'\in V(B')$, where $B'$ is a leaf-block of $C$ distinct from $B$. Apply Remark~\ref{rem:lemma18} to $B'$ to find a clean extending path $P_{xy}$ from $x\in V(B')\setminus\{w'\}$ to $y\in V(C)\setminus(V(B')\cup \{w\})$. Notice that $P_{xy}$ and $P_{v_1u_1}$ are component disjoint, otherwise there is a clean extending path $P_{v_1x}$ from $v_1$ to $x$, which is excluded by Case~\ref{case:cleanToBlock}. Also, $P_{u_0w}$ and $P_{xy}$ are also component disjoint, otherwise there is a clean extending path $P_{u_0x}$ from $u_0$ to $x$, which is excluded by Case~\ref{case:degenerateCleanToBlock}. Thus, by Remark~\ref{rem:degenerateCompDisjoint}, $P_{xy}, P_{v_1u_1}$ and $P_{u_0w}$ are pair-wise component disjoint.
    
Let $C_1$ and $C_2$ be the connected components resulting from removing the bridge $u_2w$, with $V(B)\subset V(C_1)$. If $y\in V(C_2)$ then either $y\in V(B'')$, where $B''$ is another leaf-block with $w'\in V(B'')$, and $P^C_{xy}$ contains the edges of two distinct blocks (a contradiction by Case~\ref{case:cleanToBlock}), or we are in a case symmetric to Case~\ref{case:cleanToAu1} (with $B'$ in place of $B$). So we can assume $y\in V(C_1)$ and therefore $w\notin V(P^C_{yu_1})$. Also, every path $P^C_{yu_1}$ does not use $u_0u_1$, because by exclusion of Cases~\ref{case:degenerateA0} and~\ref{case:degenerateCleanToBlock}, $y\notin A(u_0)\cup \{u_0\}$. Set $S':=S\cup E(P_{u_0w})\cup E(P_{v_1u_1})\cup E(P_{xy})\setminus\{u_0u_1\}$, so that in $S'$, $B$ and $B'$ are merged into a block $B''$ of $S'$ through the paths $P_{u_0w}\cup \{ww'\}$ and $P_{xy}\cup P^C_{yu_1}\cup P_{v_1u_1}$. Observe that the path $P^C_{wu_1}$ is also part of $B''$ in $S'$. One has
\begin{align*}
    \cost(S)-\cost(S') &\geq \credit(B)+\credit(B')+\credit(ww')+\credit(wu_2)+\credit(u_2u_1)+\credit(u_1u_0) + 1\\
    &-\credit(B'')-\widetilde\cost(P_{u_0w})-\widetilde\cost(P_{xy})-\widetilde\cost(P_{v_1u_1})\\
    &\geq 1+1+4\cdot 1/4+1-1-1-1-1=0.
\end{align*}

\begin{figure}
    \centering
    \includegraphics[scale=0.8]{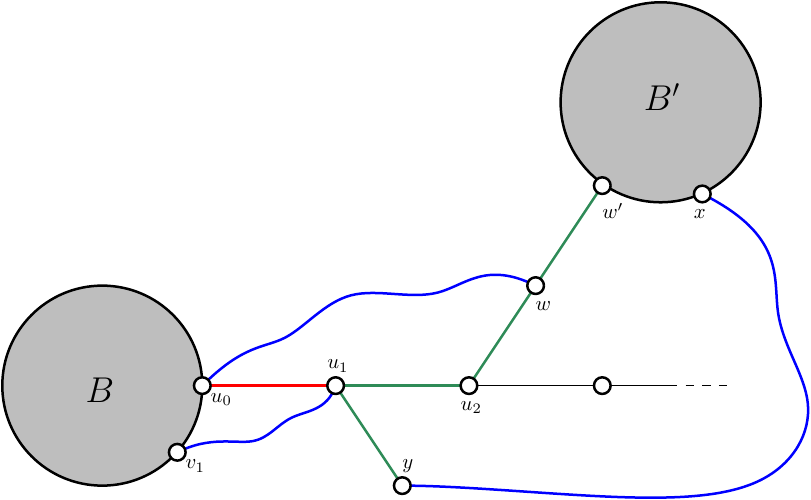}
    \caption{Lemma~\ref{lem:bridgeCovering} Case~\ref{case:degenerateAu2bridge}.}
    \label{fig:blockExtendingCasewu21}
\end{figure}

\subsubcase{There is no bridge $ww', w'\neq u_2$.}\label{case:degenerateCleanToAu2nobridge}\hfill\\
This case is illustrated in Figure~\ref{fig:blockExtendingCasewu22}. Using the fact that $P^*$ is a longest path and Remark \ref{rem:atLeastOneBridge}, it must be that $w\in V(B')$ for some leaf-block $B'$ of $C$ distinct from $B$. Apply Remark \ref{rem:lemma18} to $B'$ to find a clean extending path $P_{xy}$ from $x\in V(B')\setminus \{w\}$ to $y\in V(C)\setminus V(B')$. By the same arguments as in the previous case, $P_{v_1u_1}, P_{u_0w}$ and $P_{xy}$ are pair-wise component disjoint.
    
If $wu_2\in E(P^C_{yu_1})$ for some path $P^C_{yu_1}$, it must be that $y\in V(B'')$, where $B''$ is another leaf-block of $C$ with $w\in V(B'')$, a contradiction by Case~\ref{case:cleanToBlock}. Thus, one has $wu_2\notin P^C_{yu_1}$. Also, by exclusion of Cases~\ref{case:degenerateA0} and~\ref{case:degenerateCleanToBlock}, $u_0u_1\notin P^C_{yu_1}$. Set $S':=S\cup E(P_{v_1u_1})\cup E(P_{u_0w})\cup E(P_{xy})\setminus\{u_0u_1, u_2w\}$, so that $B$ and $B'$ are merged into a single block $B''$ of $S'$ through the paths $P_{u_0w}$ and $P_{xy}\cup P^C_{yu_1}\cup P_{v_1u_1}$. Then we get  
\begin{align*}
    \cost(S)-\cost(S') &\geq \credit(B)+\credit(B')+2-\credit(B'')-\widetilde\cost(P_{v_1u_1})-\widetilde\cost(P_{u_0w})-\widetilde\cost(P_{xy})\\
    & \geq 1+1+2-1-1-1-1=0.
\end{align*}

\begin{cfigure}
    \centering
    \includegraphics[scale=0.8]{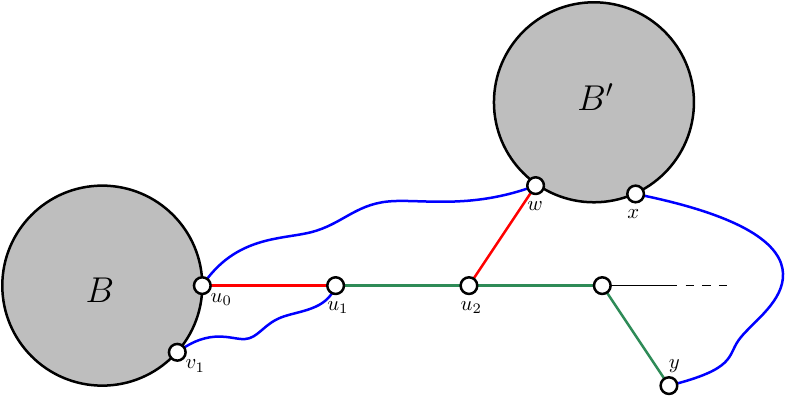}
    \caption{Lemma~\ref{lem:bridgeCovering} Case~\ref{case:degenerateCleanToAu2nobridge}.}
    \label{fig:blockExtendingCasewu22}
\end{cfigure}

\subcase{None of the previous cases hold.}\label{case:degenerateLast}\hfill\\
Let us first show that there must exist a clean extending path $P_{u_0w_1}$ from $u_0$ to $w_1\in \{u_2, u_3\}\setminus\{w_2\}$. Recall that $\{e_1, e_2, e_3\}$ is a matching in $G_C$ such that $e_1$ and $e_2$ correspond to $P_{v_1u_1}$ and $P_{v_2w_2}$, respectively. Recall also that $w_2\in\{u_2, u_3\}$. Notice that $e_3$ must be incident to $u_0$ in $B$, otherwise it is incident to some $v_3\in V(B)\setminus\{u_0, v_1, v_2\}$ in $B$, and thus it corresponds to some extending path $P_{v_3w_1}$ from $v_3$ to $w_1\in V(C)\setminus(V(B)\cup\{u_1, w_2\})$. By exclusion of Case~\ref{case:nonCleanPath}, $P_{v_3w_1}$ is clean, and by Remark~\ref{rem:extendingPathEndpoints}, it must be that $w_1\in\{u_1, u_2, u_3\}$. Thus, the paths $P_{v_2w_2}, P_{v_3w_1}$ imply a contradiction to the fact that Case~\ref{case:twoPaths} does not hold. Therefore, $e_3$ is incident to $u_0$ in $B$. Notice that, since $e_1$ is incident to $u_1$, it must be that $e_3$ is not incident to $u_1$. But then, since $u_0u_1$ is the only bridge of $C$ incident to $V(B)$ and Case~\ref{case:degenerateA0} does not hold, the only edge of $C$ not in $B$ incident to $u_0$ is $u_0u_1$, and thus it must be that the edge $e_3$ corresponds to some extending path $P_{u_0w_1}$ from $u_0$ to some node $w_1\in V(C)\setminus(V(B)\cup\{u_1, w_2\})$.

Now, since Case~\ref{case:degenerateNonClean} does not hold, $P_{u_0w_1}$ is clean. Since Cases~\ref{case:degenerateA0},~\ref{case:degenerateAu1} and~\ref{case:degenerateAu2} do not hold, $w_1\in V(C)\setminus (A(u_0)\cup A(u_1)\cup A(u_2)\cup \{u_0, u_1\})$. If $w_1\in V(C)\setminus (A(u_0)\cup A(u_1)\cup A(u_2)\cup \{u_0, u_1, u_2, u_3\})$ then $|E(P^C_{wu_0})|\geq 4$, so that Case~\ref{case:degenerateCleanToBlock} holds, a contradiction. Thus $w_1\in \{u_2, u_3\}\setminus\{w_2\}$.

We remark that $P_{u_0w_1}$ and $P_{v_2w_2}$ might not be component disjoint. Removing the bridge $u_2u_3$ splits $C$ into $2$ connected components $C_1, C_2$, with $V(B)\subset V(C_1)$.

\subsubcase{There is a non-clean extending path $P_{ab}$ from $a\in V(C_1)$ to $b\in V(C_2)$.}\label{case:degeneratetwoPathsNonClean}\hfill\\
This case is illustrated in Figure~\ref{fig:blockExtendingcasewu31}. Let $C'$ be the $4$-cycle corresponding to $P_{ab}$ and let $aa'$ and $bb'$ be the first and last edges of $P_{ab}$. Apply the 3-matching Lemma~\ref{lem:matchingOfSize3} to $C'$ to get a matching $\{x_1y_1, x_2y_2, x_3y_3\}$, with $x_i\in C, y_i\in C'$ for all $i$. By the same arguments as in Case~\ref{case:twoPathsNonClean}, we can assume w.l.o.g. that $x_1\in V(C_1)$, $x_2\in V(C_2)$ and $y_1y_2\in E(C')$.
    
If $w_2 = u_3$, then the same construction and analysis as in Case~\ref{case:twoPathsNonClean} holds. Otherwise, one has $w_1 = u_3$. By Remark~\ref{rem:degenerateCompDisjoint}, the paths $P_{v_1u_1}$ and $P_{u_0u_3}$ are component disjoint. Set $S':=S\cup E(P_{v_1u_1})\cup E(P_{u_0u_3})\cup\{x_1y_1, x_2y_2\}\setminus\{u_0u_1, y_1y_2\}$. Notice that since $x_1\in V(C_1), x_2\in V(C_2)$, there are vertex disjoint paths $P^C_{x_1u_2}$ and $P^C_{x_2u_3}$. Moreover, by Case~\ref{case:degenerateA0}, none of those paths contains the edge $u_0u_1$. Notice $P^C_{u_1u_3}=u_1u_2u_3$. Adding $P_{v_1u_1}$ and $P_{u_0u_3}$ and removing $u_0u_1$ merges $B$ and $P_{v_1u_1}\cup P^C_{u_0u_3}$ into a single block, and adding $E(C')\cup\{x_1y_1, x_2y_2\}\setminus\{y_1y_2\}$ merges that block and the paths $P^C_{x_1u_2}$, $P^C_{x_2u_3}$ into a new block $B'$ of $S'$. One has
\begin{align*}
    \cost(S)-\cost(S') &\geq \credit(B)+\credit(C') +\credit(u_0u_1)+\credit(u_1u_2)+\credit(u_2u_3)\\
    &- \credit(B') - \widetilde\cost(P_{v_1u_1}) -\widetilde\cost(P_{u_0u_3})\geq 1 + 4/3 + 3\cdot 1/4 - 1 - 1 - 1 >0.
\end{align*}

\begin{figure}
    \centering
    \includegraphics[scale=0.8]{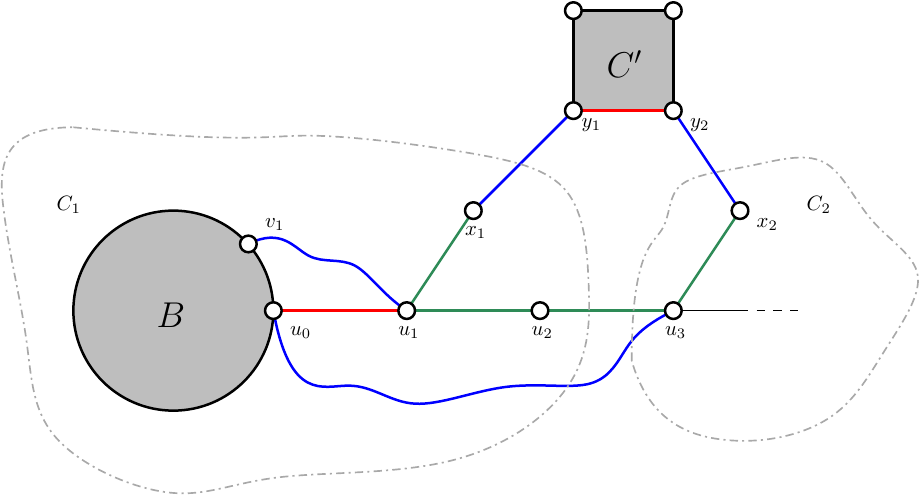}
    \caption{Lemma~\ref{lem:bridgeCovering} Case~\ref{case:degeneratetwoPathsNonClean}.}
    \label{fig:blockExtendingcasewu31}
\end{figure}

\subsubcase{All extending paths between nodes of $C_1$ and $C_2$ are clean.}\label{case:degenerateLastClean}\hfill\\
This case is illustrated in Figure~\ref{fig:blockExtendingcasewu3}. Apply Lemma~\ref{lem:3matchingPath} to $(V(C_1), V(C_2))$. By the hypothesis of this case, there must be at least one clean extending path $P_{xy}$ from $x\in V(C_2)\setminus\{u_3\}$ to $y\in V(C_1)\setminus\{u_0\}$. Since $x\neq u_3$, $P^C_{xu_3}$ is not empty, and it contains at least one edge $e$. Thus, $|E(P^C_{xu_0})|\geq 4$, so $y\notin A(u_0)$ by Case~\ref{case:cleanToBlock}, and thus $y\in A(u_1)\cup A(u_2)\cup\{u_1, u_2\}$. Observe that the paths $P_{v_1u_1}$ and $P_{xy}$ are pair-wise component disjoint. Indeed, if $P_{v_1u_1}$ and $P_{xy}$ are not component disjoint there is a clean extending path $P_{v_1x}$ from $v_1$ to $x$, and there is a path $P^C_{v_1x}$ that includes an edge of $B$ and at least $4$ other edges not in $B$ (since $|E(P^C_{xu_0})|\geq 4$); this is excluded by Case~\ref{case:cleanToBlock}. By a similar argument $P_{xy}$ and $P_{v_2w_2}$ are also component disjoint. Also, the paths $P_{u_0w_1}$ and $P_{xy}$ must also be component disjoint, otherwise there is a clean extending path $P_{u_0x}$ from $u_0$ to $x$, and there is a path $P^C_{u_0x}$ that includes an edge of $B$ and at least $4$ other edges not in $B$ (since $|E(P^C_{xu_0})|\geq 4$); this is excluded by Case~\ref{case:degenerateCleanToBlock}. Finally, $P_{v_1u_1}$ and $P_{u_0w_1}$ are also component disjoint by Remark~\ref{rem:degenerateCompDisjoint}.
    
Assume first $y\in A(u_1)\cup\{u_1\}$. If $w_2=u_2$ then the same construction and analysis as in Case~\ref{case:twoPathsClean}, with $P_{v_2w_2}$ in place of $P_{v_1u_2}$ holds. Otherwise, one has $w_1 = u_2$. Set $S':=S\cup E(P_{v_1u_1})\cup E(P_{u_0u_2})\cup E(P_{xy})\setminus \{u_0u_1, u_1u_2\}$. Notice that the edges $u_0u_1, u_1u_2$ are not in $P^C_{xu_3}$. Also, the edges $u_0u_1, u_1u_2$ are not in $P^C_{u_1y}$ because $y\in A(u_1)\cup\{u_1\}$. The path $P_{v_1u_1}\cup P^C_{u_1y}\cup P_{xy}\cup P^C_{xu_3}\cup\{u_2u_3\}\cup P_{u_0u_2}$ is merged with $B$ into a single block $B'$ of $S'$. If $e$ is the edge of a block then that block is also merged into $B'$, bringing $1$ credit, and if it is a bridge it brings $1/4$ credits, so we get at least $1/4$ credits from the edges of $P^C_{xu_3}$. One has
\begin{align*}
    \cost(S)-\cost(S') &\geq \credit(B)+\credit(u_0u_1)+\credit(u_1u_2)+\credit(u_2u_3)+1/4+2\\
    &-\credit(B')-\widetilde\cost(P_{u_1v_1})-\widetilde\cost(P_{u_0u_2})-\widetilde\cost(P_{xy})\\
    &\geq 1+4\cdot 1/4+2-1-1-1-1=0.
\end{align*}
Assume now that $y\in A(u_2)\cup\{u_2\}$. If $w_2=u_3$ then the same construction and analysis as in Case~\ref{case:twoPathsClean}, with $P_{v_2w_2}$ in place of $P_{v_2u_3}$ holds. Otherwise, one has $w_1 = u_3$. Set $S':=S\cup E(P_{v_1u_1})\cup E(P_{u_0u_3})\cup E(P_{xy})\setminus \{u_0u_1, u_2u_3\}$. By similar arguments as above, the edges $u_0u_1, u_2u_3$ are neither in $P^C_{xu_3}$ nor in $P^C_{u_2y}$. The path $P_{v_1u_1}\cup \{u_1u_2\}\cup P^C_{u_2y}\cup P_{xy} \cup P^C_{xu_3}\cup P_{u_0u_3}$ is merged with $B$ into a single block $B'$ of $S'$. As before, if $e$ is the edge of a block then that block is also merged into $B'$, bringing $1$ credit, and if it is a bridge it brings $1/4$ credits, so we get at least $1/4$ credits from the edges of $P^C_{xu_3}$. One has 
    \begin{align*}
        \cost(S)-\cost(S') &\geq \credit(B)+\credit(u_0u_1)+\credit(u_1u_2)+\credit(u_2u_3)+1/4+2\\
        &-\credit(B')-\widetilde\cost(P_{v_1u_1})-\widetilde\cost(P_{u_0u_3})-\widetilde\cost(P_{xy})\\
        &\geq 1+4\cdot 1/4+2-1-1-1-1=0.
    \end{align*}

\begin{figure}
    \centering
    \includegraphics[scale=0.8]{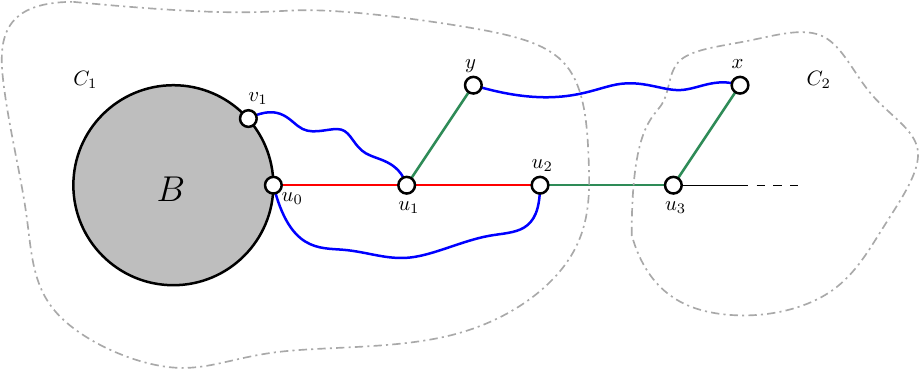}
    \includegraphics[scale=0.8]{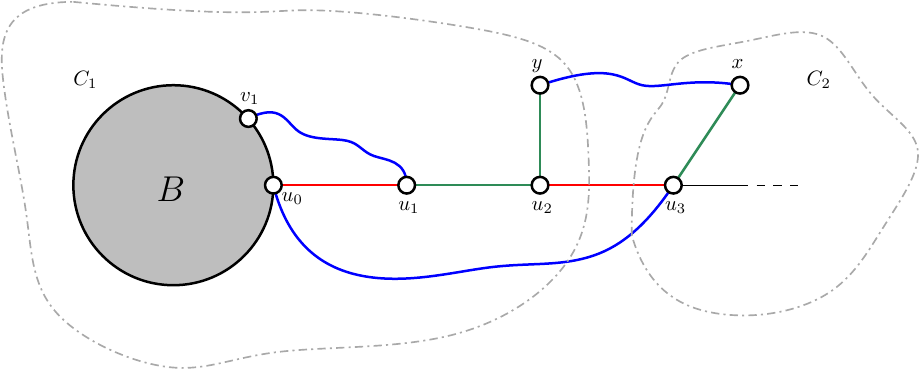}
    \caption{Lemma~\ref{lem:bridgeCovering} Case~\ref{case:degenerateLastClean}.}
    \label{fig:blockExtendingcasewu3}
\end{figure}

\end{caseanalysis}

Notice that in each case the constructed $S'$ is canonical. Indeed, we never create connected components with fewer than $6$ edges, nor any leaf-block with fewer than $5$ nodes. Furthermore, each $4$-cycle which is adjacent to a unique other large component $C'$ preserves this property (though $C'$ might turn from 2VC to complex or vice versa). Finally, the number of complex components decreases since we iteratively apply the above construction to a complex component $C$ until it has only one block and thus is no longer a complex component.
\end{proof}

\printbibliography

\end{document}